\documentclass[reqno,11pt]{amsart}
\usepackage[utf8]{inputenc}

\usepackage{amsmath}
\usepackage{amsfonts}
\usepackage{amssymb}
\usepackage{amsthm}
\usepackage{breqn}
\usepackage{latexsym}
\usepackage{hyperref} 
\usepackage{physics}
\usepackage{graphicx}
\usepackage[dvipsnames]{xcolor}
\usepackage{verbatim}
\usepackage{caption}
\usepackage{subcaption}
\usepackage{tikz}

\theoremstyle{definition}
\newtheorem{defn}{Definition}[section]

\theoremstyle{definition}
\newtheorem{prop}{Proposition}[section]

\theoremstyle{theorem}
\newtheorem{thm}{Theorem}[section]

\theoremstyle{theorem}
\newtheorem{lemma}{Lemma}[section]

\theoremstyle{remark}
\newtheorem{rem}{Remark}[section]

\theoremstyle{hypotesis}

\theoremstyle{hypotesis}

\newcommand{\bigslant}[2]{{\raisebox{.2em}{$#1$}\left/\raisebox{-.2em}{$#2$}\right.}}

\title[\tiny{Propagation of singularities on asymptotically AdS spacetimes}]{A generalization of the propagation of singularities theorem on asymptotically anti-de Sitter spacetimes}

\author[C.\ Dappiaggi]{Claudio Dappiaggi}
\address{Dipartimento di Fisica \\ Universit{\`a} degli Studi di Pavia\\ Istituto Nazionale di Alta Matematica, Sezione di Pavia and INFN, Sezione di Pavia \\ Via Bassi, 6 --  I-27100 Pavia \\ Italy}
\email{claudio.dappiaggi@unipv.it}

\author[A.\ Marta]{Alessio Marta}
\address{Dipartimento di Matematica \\ Universit{\`a} degli Studi di Milano\\ and INFN, Sezione di Milano \\ Via Cesare Saldini, 50 --  I-20133 Milano \\ Italy}
\email{alessio.marta@unimi.it}

\date{\today}

\begin{document}
	
	\maketitle

\begin{abstract}
In a recent paper O. Gannot and M. Wrochna considered the Klein-Gordon equation on an asymptotically anti-de Sitter spacetime subject to Robin boundary conditions, proving in particular a propagation of singularity theorem. In this work we generalize their result considering a more general class of boundary conditions implemented on the conformal boundary via pseudodifferential operators of suitable order. Using techniques proper of $b$-calculus  and of twisted Sobolev spaces, we prove also for the case in hand a propagation of singularity theorem along generalized broken bicharacteristics, highlighting the potential presence of a contribution due to the pseudodifferential operator encoding the boundary condition. 
\end{abstract}

\section{Introduction}\label{Sec: Introduction}

In the framework of Lorentzian geometry, the $n$-dimensional asymptotically anti-de Sitter (aAdS) spacetimes play an important role since they represent a distinguished class of manifolds admitting a conformal boundary endowed with an induced Lorentzian metric. The main representative of this family is anti-de Sitter $AdS_n$, which is a maximally symmetric solution of Einstein's equations with negative cosmological constant. On top of these backgrounds it is natural to consider the Klein-Gordon operator and its properties have been studied by several authors, {\it e.g.} \cite{Bachelot:2010zw,Enciso:2013lza,Hol12,War1,Vasy12}, to quote a few papers which have been of inspiration to this work. 

One of the interesting aspects of aAdS spacetimes concerns the fact that, contrary to globally hyperbolic spacetimes, in order to solve the Klein-Gordon equation, besides initial data it is necessary to impose boundary conditions. Those of Dirichlet type have caught the interest for several years and only recently, in the mathematical physics literature, the attention has been moved towards other choices. The first natural generalization consists of considering boundary conditions of Robin type, as discussed for example in \cite{War1} and in \cite{Dappiaggi:2017wvj, WaIs03} within the framework of quantum field theory. Especially in this context the main objects of interest are the fundamental solutions of the Klein-Gordon operator, particularly the advanced and retarded ones, as well the propagation of singularities. For Dirichlet boundary conditions the latter has been studied by Vasy in \cite{Vasy12}, while for the Robin case the problem was addressed recently by Gannot and Wrochna in \cite{GaWr18}, proving in addition the existence and uniqueness up to smoothing terms of parametrices for the Klein-Gordon operator with prescribed $b$-wavefront set. 

Yet, in the past few years, it has emerged that one could consider a larger class of boundary conditions, in between which a distinguished example are those of Wentzell type as advocated in the realm of quantum field theory in \cite{Zahn:2015due}, although they have been considered by several other authors in different contexts, {\it e.g.} \cite{Coclite:2014,Favini:2002,Ueno:1973}. In a few words, given an aAdS spacetime $M$, and given $u\in H^1_{loc}(M)$, it is possible to define two trace maps $\gamma_+$ and $\gamma_-$ the first encoding the Neumann data, the second the Dirichlet ones. While Robin boundary conditions are codified by a smooth function $f$ on $\partial M$ such that $\gamma_+u-f\gamma_-u=0$, in those of Wentzell type the role of $f$ is replaced by suitable second order differential operator acting on the boundary. From the viewpoint of applications, similarly to those of Robin type -- see \cite{GaWr18} but also \cite{Dybalski:2018egv}, these conditions are relevant in connection to the so-called holographic principle as advocated in \cite{Dappiaggi:2018pju,Zahn:2015due}. 

From a structural viewpoint it has been shown in \cite{DDF18}, using the notion of boundary triples, that there exists a large class of boundary conditions relating Neumann and Dirichlet data via pseudodifferential operators for which there exist advanced and retarded fundamental solutions for the Klein-Gordon operator. It is noteworthy that Wentzell boundary conditions fall in this class, although the analysis makes clear that this is just one of the many possible scenarios. Yet one of the key limitations of \cite{DDF18} is the lack of a complete control of the wavefront set of the propagators, mainly due to the lack of a theorem of propagation of singularities applicable to such scenario. 

In this work we shall bypass such limitation proving a theorem of propagation of singularities for the Klein-Gordon operator on an asymptotically anti-de Sitter spacetime $M$ such that the boundary condition is implemented by a $b$-pseudodifferential operator $\Theta\in\Psi^k_b(\partial M)$ with $k\leq 2$. As it will become manifest from out analysis, we can distinguish two notable cases, namely $k\leq 0$ and $0<k\leq 2$. The first one can be seen as a rather natural extension of the results of \cite{GaWr18}, while the second accounts for the main, novel interesting cases and, in particular, Wentzell boundary conditions fall in this class. 

In our endeavor we shall follow the same strategy and techniques adopted first by Vasy in \cite{Vasy12} and subsequently by Wrochna and Gannot in \cite{GaWr18}, when dealing with boundary conditions of Robin type. As the subscript $b$ suggests, we shall mainly use techniques proper of $b$-calculus, $b$-wavefront sets and of twisted Sobolev spaces, which were first introduced in this framework by Warnick in \cite{War1}. More precisely we shall prove two main theorems, {\it cf.} Theorem \ref{Thm: main theorem k positivo} and \ref{Thm: main theorem k negativo}. While the latter can be seen as a natural extension of \cite[Th. 1]{GaWr18} and it accounts for the case $k\leq 0$ mentioned above, the second one deals with $0<k\leq 2$. Most notably this scenario opens the possibility for the boundary conditions to yield an additional contribution to the underlying wavefront set contrary to the first case.

\vskip .3cm

The paper is organized as follows: In Section \ref{Sec: Geometric preliminaries} we discuss the key geometric ingredients necessary in our work. We start with an introduction to $b$-geometry in Section \ref{Sec: Geometric preliminaries} and we discuss the notions of globally hyperbolic spacetimes with timelike boundary and of manifolds of bounded geometry in Section \ref{Sec: Globally Hyperbolic} and \ref{Sec: Manifolds of bounded geometry} respectively. In Section \ref{Sec: Analytic Preliminaries} we introduce the analytic tools necessary in this paper. We start from a survey on $b$-pseudodifferential operators in Section \ref{Sec: bPsiDos} and on twisted Sobolev spaces in \ref{subsec:twisted_sobolev}. Here we spend some time in motivating their introduction focusing on a simple, yet in our opinion enlightening example. Subsequently we discuss the interplay between b-calculus and wavefront set in Section \ref{Sec: b-calculus and WF}, while in Section \ref{Sec: Asymptotic expansion} we introduce the relevant trace maps needed to discuss the boundary conditions. At last in Section \ref{Sec: Twisted Dirichlet form} we introduce one of the main key ingredients of the work: the twisted Dirichlet energy form. In Section \ref{Sec: Boundary Value Problem} we give a weak formulation of the problem we are interested in in Section \ref{Sec: Boundary conditions and the associated Dirichlet form} proving in particular some microlocal estimates for the associated Dirichlet form. Section \ref{Sec: propagation of singularities} represents the core of the paper. We start by introducing the key notion of compressed  characteristic set and of generalized broken bicharacteristic in Section \ref{Sec: compressed characteristic set}. Herein we identify three notable conic subsets of the $b$-cotangent bundle, the elliptic, the hyperbolic and the glancing regions. For each of these we need to prove suitable microlocal estimates which are discussed respectively in Section \ref{Sec: elliptic region}, \ref{Sec: hyperbolic region} and \ref{Sec: glancing region}. Finally we gather all these data to derive the sought theorem of propagation of singularities in \ref{Sec: propagation of singularities}. In the whole analysis we separate two cases, that in which the boundary condition is implemented by $\Theta\in\Psi^k_b(\partial M)$, $k\leq 0$ and that in which $0<k\leq 2$.

\section{Geometric preliminaries}\label{Sec: Geometric preliminaries}
In this section our main goal is both to fix notation and conventions and to introduce the three geometric concepts which play a key role in our analysis: $b$-geometry, globally hyperbolic spacetimes with timelike boundary and manifolds of bounded geometry.

\subsection{Introduction to b-geometry} We introduce and characterize suitable geometric structures which are the natural playground to discuss the propagation of singularities on manifolds with boundaries. The concept of {\em b-geometry} has been first introduced by R. Melrose in \cite{Mel93,MePi92} and its use has been advocated by several authors, see in particular \cite{GaWr18,Vasy08}. Here we give a slightly different and more general version of the key ingredients of b-geometry which is inspired by the presentation in \cite{GMP}.

In this section with $M$ we indicate a connected, orientable, smooth manifold of dimension $\dim M=n\geq 2$, possibly with non empty boundary. Let $S$ denote a smooth submanifold of $M$ of dimension $\dim S=n-1$, so that the natural map $\iota:S\hookrightarrow M$ is an injective immersion and an homeomorphism on its image. In addition we call $NS$ the rank $1$ normal bundle associated to $S$.

Under these assumptions, one can apply the {\em tubular neighbourhood theorem} to conclude that there exists both an open subset $V$ of the zero section $\sigma_0:\iota(S)\to NS$ and $U\subset M$ such that $\iota(S)\subset U$. In addition there exists a diffeomorphism $\psi:V\to U$ such that
$$\psi\circ\sigma_0=\iota.$$

For future convenience we introduce the following space of smooth sections and we adopt the same notation employed in \cite{GaWr18} to facilitate a comparison. Hence we call  $\mathcal{V}(M)\doteq\Gamma(TM)$, $\mathcal{V}(S)=\Gamma(TS)$ and
\begin{equation}\label{eq:Notation}
\mathcal{V}_S(M)\doteq\{X\in\Gamma(TM)\;|\; X|_S\in\Gamma(TS)\},
\end{equation}
where, with a slight abuse of notation, with $X|_S\in\Gamma(TS)$ we mean that the restriction of $X$ to $S$ is tangent to the submanifold. Our goal is to characterize $\mathcal{V}_S(M)$ as the space of smooth sections of a suitable vector bundle. To this end, we observe first of all, that, for every open neighbourhood $\widetilde{U}\subset M$ such that $\widetilde{U}\cap\iota(S)=\emptyset$, $\mathcal{V}_S(\widetilde{U})=\{X|_{\widetilde{U}}\;|\;X\in\mathcal{V}_S(M)\}$ coincides with $\mathcal{V}(\widetilde{U})$. 

For a complete characterization of \eqref{eq:Notation}, we apply the tubular neighbourhood theorem, to conclude that there must exist an open neighbourhood $U_\epsilon$ of $\iota(S)$ diffeomorphic to $(-\epsilon,\epsilon)\times\iota(S)$, $\epsilon>0$. Calling $x$ the coordinate on the interval $(-\epsilon,\epsilon)$, we can realize that, given any $X\in\mathcal{V}_S(M)$, then
$$X|_{U_\epsilon}=f\partial_x + Y_x,$$
where $f\in C^\infty(U_\epsilon)$ is such that $f|_{x=0}=0$ while $Y_x\in\Gamma(TU_\epsilon)$ for all $z\in(-\epsilon,\epsilon)$ with the constraint that $Y_0\in\Gamma(TS)$. Since $f$ is smooth and vanishing at $x=0$, it holds that, there exists $\alpha\in C^\infty(U_\epsilon)$ such that $f=x\alpha$. 

Collecting all these data, we are motivated to introducing a new bundle $^STM$ whose base space is $M$ and whose fiber is defined as follows:
\begin{equation}\label{Eq:fiber}
^ST_pM\doteq\left\{\begin{array}{ll}
T_pM & p\notin U_\epsilon\\
\textrm{span}_{\mathbb{R}}\left(x\partial_x, T_pS\right) & p \in U_\epsilon
\end{array}\right.
\end{equation}
We observe that the definition does not depend on the choice of $\epsilon$, since, if $p\notin\iota(S)$, then one can always find an open neighbourhood $U_p\subset M$ containing $p$ and not intersecting $S$ such that $^STM|_{U_p}$ is diffeomorphic to $TM|_{U_p}$. The above analysis can now be summarized in the following proposition:
\begin{prop}\label{Prop:S-sections}
	The vector space $\mathcal{V}_S(M)$, defined in \eqref{eq:Notation}, is isomorphic to $\Gamma(^STM)$. 
\end{prop}
\noindent Observe that the natural restriction map 
$$\pi_S:\mathcal{V}_S(M)\to\Gamma(TS)\quad X\mapsto X|_S$$
is not injective and its kernel can be characterized as a line bundle over $S$ with a canonical non trivial section, {\it cf.} \cite[Prop. 4]{GMP}. In order to make contact with \cite{Mel92} and with the notation introduced therein, we stress that Proposition \ref{Prop:S-sections} and the preceding discussion applies also when $M$ has not empty boundary and $S=\partial M$. The only difference lies in the form of the tubular neighbourhood, namely $(-\epsilon,\epsilon)$ ought to be replaced by $[0,\epsilon)$. In this case 
\begin{defn}
	Let $M$ be a connected, orientable smooth manifold with $\partial M\neq\emptyset$. We call $b$-tangent bundle $^bTM\doteq\, ^{\partial M}TM$ as per equation \eqref{Eq:fiber}. 
\end{defn}

In the following we shall refer to $^bT^*M$ as the $b$-cotangent bundle which is a finite rank vector bundle over $M$ dual to $^bTM$. We remark that, for all $p\in M\setminus\partial M$, $^bT^*_pM$ coincides with $T^*_pM$, while, if $p\in\partial M$, $^bT^*_pM=\textrm{span}_{\mathbb{R}}\left\{T^*_p\partial M,\frac{dx}{x}\right\}$. In addition one can observe that there exists a natural non injective map $\pi:T^*M\to{}^bT^*M$, built as follows. Let us consider a tubular neighbourhood of $\partial M$ and a chart $U$ centered at point $p\in\partial M$. Hereon we can consider local coordinates $(x,y_i,\xi,\eta_i)$ of $T^*M$, $i=1,\dots,\dim\partial M$ as well as local coordinates $(x,y_i,\zeta,\eta_i)$ of $T^*_bM$. The projection map $\pi$ acts as follows:
$$\pi(x,y_i,\xi,\eta_i)=(x,y_i,x\xi,\eta_i).$$
On the contrary if we consider a chart $U^\prime$ centered at a point $q\in\mathring{M}=M\setminus\partial M$, thereon the map $\pi$ is nothing but the identity. Hence, one can realize that $\pi\in C^\infty(T^*M;{}^bT^*M)$, though it is not injective. We call {\em compressed $b$-cotangent bundle} 
\begin{equation}\label{Eq: compressed b-cotangent bundle}
^b\dot{T}^*M\doteq \pi[T^*M],
\end{equation}
which is a subset of $^bT^*M$. Further details can be found in \cite{Mel81, Vasy12}. The last geometric structure that we shall need in this work is the {\em b-cosphere bundle} which is realized as the quotient manifold obtained via the action of the dilation group on $T^*_bM\setminus\{0\}$, namely
\begin{equation}\label{Eq: cosphere bundle}
{}^bS^*M\doteq\bigslant{{}^bT^*M\setminus\{0\}}{\mathbb{R}^+}.
\end{equation}
We remark that, if we consider a local chart $U\subset M$ such that $U\cap\partial M\neq\emptyset$ and the local coordinates $(x,y_i,\zeta,\eta_i)$, $i=1,\dots,n-1=\dim\partial M$, on ${}^bT^*_UM\doteq{}^bT^*M|_U$, we can build a natural counterpart on ${}^bS^*_UM$, namely $(x,y_i,\widehat{\zeta},\widehat{\eta}_i)$ where $\widehat{\zeta}=\frac{\zeta}{\rho}$ and $\widehat{\eta_i}=\frac{\eta_i}{\rho}$ with $\rho=|\eta_{n-1}|$.

To conclude the section, we introduce a class of differential operators which is naturally built out of $\mathcal{V}_{\partial M}(M)$.

\begin{defn}\label{Def: b-differential operators}
Let $M$ be a connected, orientable smooth manifold with $\partial M\neq\emptyset$. We call $Diff_b(M)=\oplus_{k=0}^\infty Diff^k_b(M)$ the graded differential operator algebra generated by $\mathcal{V}_{\partial M}(M)$.
\end{defn}

\subsection{Globally hyperbolic spacetimes with timelike boundary}\label{Sec: Globally Hyperbolic}

In this section we specify a distinguished class of manifolds which play a key role in our construction, since they are the natural playground where one can expect that the mixed initial/boundary value problem for partial differential equations ruled by a normally hyperbolic operator is well-posed. Our analysis summarizes the main results obtained by \cite{AFS18}. We also assume that the reader is acquainted with the basic notions of Lorentzian geometry, {\it e.g.} \cite{ONeill83}.

\begin{defn}\label{Def: globally hyperbolic}
	Let $(M,g)$ be a connected, oriented, time oriented, smooth Lorentzian manifold of dimension $\dim M=n\geq 2$ with non empty boundary $\iota:\partial M\to M$. We say that $(M,g)$ 
	\begin{enumerate}
		\item has a {\bf timelike boundary} if $(\partial M,\iota^*g)$ identifies a smooth, Lorentzian manifold,
		\item is {\bf globally hyperbolic} if it does not contain any closed causal curve and if, for every $p,q\in M$, $J^+(p)\cap J^-(q)$ is either empty or compact. Here $J^\pm$ stand for the causal future (+) and past (-).
	\end{enumerate}
If both conditions are met, we call $(M,g)$ a {\em globally hyperbolic spacetime with timelike boundary}.
\end{defn}

Observe that, for simplicity, we assume throughout the paper that also $\partial M$ is {\em connected}. The following theorem, proven in \cite{AFS18}, gives a more explicit characterization of the class of manifolds, we are interested in. As a preliminary step, we recall that, given a Lorentzian manifold $(M,g)$, a Cauchy surface $\Sigma$ is an achronal subset of $M$ such that every inextensible, piecewise smooth, timelike curve intersects $\Sigma$ only once.

\begin{thm}\label{Th: globally hyperbolic}
	Let $(M,g)$ be a globally hyperbolic spacetime with timelike boundary of dimension $\dim M=n\geq 2$. Then it is isometric to a Cartesian product $\mathbb{R}\times\Sigma$ where $\Sigma$ is an $(n-1)$-dimensional Riemannian manifold. The associated line element reads 
	$$ds^2=-\beta d\tau^2 + \kappa_\tau,$$
	where $\beta\in C^\infty(\mathbb{R}\times\Sigma;(0,\infty))$ while $\tau:\mathbb{R}\times\Sigma\to\mathbb{R}$ plays the role of time coordinate. In addition $\mathbb{R}\ni\tau\mapsto\kappa_{\tau}$ identifies a family of Riemmannian metrics, smoothly dependent on $\tau$ and such that, calling $\Sigma_\tau\doteq\{\tau\}\times\Sigma$, each $(\Sigma_\tau,\kappa_\tau)$ is a Cauchy surface with non empty boundary.
\end{thm}

\begin{rem}\label{Rem: boundary metric}
Observe that a notable consequence of this theorem is that, calling $\iota_{\partial M}:\partial M\to M$ the natural embedding map, then $(\partial M,h)$ where $h=\iota^*_{\partial M}g$ is a globally hyperbolic spacetime. In particular the associated line element reads 
$$ds^2|_{\partial M}=-\beta|_{\partial M}d\tau^2+\kappa_{\tau}|_{\partial M}.$$
\end{rem}

\subsection{Asymptotically anti-de Sitter spacetimes}

In this subsection we recall briefly the class of backgrounds which have been considered in \cite{GaWr18} and which represents a key ingredient also in our investigation.
	
\begin{defn}\label{Def: asymptotically AdS}
Let $M$ be an n-dimensional manifold with non empty boundary $\partial M$. Suppose that $\mathring{M}=M\setminus\partial M$ is equipped with a smooth Lorentzian metric $g$ and that 
\begin{itemize}
\item[a)] If $x \in \mathcal{C}^\infty(M)$ is a boundary function, then $\widehat{g} = x^2 g$ extends smoothly to a Lorentzian metric on $M$.
\item[b)] The pullback $h=\iota^*_{\partial M}\widehat{g}$ via the natural embedding map $\iota_{\partial M}:\partial M\to M$ individuates a smooth Lorentzian metric.
\item[c)] $\widehat{g}^{-1}(dx,dx)=1$ on $\partial M$.
\end{itemize}
Then $(M,g)$ is called an {\em asymptotically anti-de Sitter (AdS) spacetime}. In addition, if $(M,\widehat{g})$ is a globally hyperbolic spacetime with timelike boundary, {\it cf.} Definition \ref{Def: globally hyperbolic}, then we call $(M,g)$ a {\em globally hyperbolic asymptotically AdS spacetime}.
\end{defn}

Observe that conditions a), b) and c) are actually independent from the choice of the boundary function $x$ and the pullback $h$ is actually determined up to a conformal multiple since there exists always the freedom of multiplying the boundary function $x$ by any nowhere vanishing $\Omega\in C^\infty(M)$. Such freedom plays no role in our investigation and we shall not consider it further. 

As a direct consequence of the collar neighbourhood theorem and of the freedom in the choice of the boundary function in Definition \ref{Def: asymptotically AdS}, this can always be engineered in such a way, that, given any $p\in\partial M$, it is possible to find a neighbourhood $U\subset\partial M$ containing $p$ and $\epsilon>0$ such that on $U\times[0,\epsilon)$ the line element associated to $g$ reads
\begin{equation}\label{Eq: metric near the boundary}
ds^2 = \frac{-dx^2+h_x}{x^2}
\end{equation}
where $h_x$ is a family of Lorentzian metrics depending smoothly on $x$ such that $h_0\equiv h$.

\begin{rem}
It is important to stress that the notion of asymptotically AdS spacetime given in Definition \ref{Def: asymptotically AdS} is actually more general than the one given in \cite{Ashtekar:1999jx}, which is more commonly used in the general relativity and theoretical physics community. Observe in particular that $h_x$ in Equation \eqref{Eq: metric near the boundary} does not need to be an Einstein metric nor $\partial M$ is required to be diffeomorphic to $\mathbb{R}\times\mathbb{S}^{n-2}$. Since we prefer to make a close connection to \cite{GaWr18} we stick to their nomenclature.
\end{rem}

\begin{rem}\label{Rem: x as global coordinate}
	With a slight abuse of notation and in view of Definition \ref{Def: asymptotically AdS}, henceforth we shall use the symbol $x$ both when referring to the boundary function of an asymptotically AdS spacetime $(M,g)$ and when considering the coordinate normal to $\partial M$.
\end{rem}

\subsection{Manifolds of bounded geometry}\label{Sec: Manifolds of bounded geometry}
To conclude the section we introduce another notable class of manifolds namely those of bounded geometry. These play a key role in defining Sobolev spaces when the underlying background has a non empty boundary. In this section we outline these concepts in an abridged form, in order to keep this work self-consistent. An interested reader can find more details in \cite{Sch01,ANN16,GS13,GOW17} as well as in \cite[Sec. 2.1 \& 2.2]{DDF18}. 

\begin{defn}\label{def:manifold_bounded_wb}
A Riemannian manifold $(N,h)$ with empty boundary is of bounded geometry if 
\begin{itemize}
\item[a)] The injectivity radius $r_{inj}(N)$ is strictly positive,
\item[b)] $N$ is of totally bounded curvature, namely for all $k \in \mathbb{N}\cup\{0\}$ there exists a constant $C_k>0$ such that $\| \bigtriangledown^k R\|_{L^\infty(M)} < C_k$.
\end{itemize}
\end{defn}

This definition cannot be applied slavishly to a manifold with non empty boundary and, to extend it, we need to introduce a preliminary concept.

\begin{defn}\label{def:submanifold_bounded}
Let $(N,h)$ be a Riemannian manifold of bounded geometry and let $(Y,\iota_Y)$ be a codimension $k$, closed, embedded smooth submanifold with an inward pointing, unit normal vector field $\nu_Y$. The submanifold $(Y,\iota^*_Y g)$ is of bounded geometry if:
\begin{itemize}
\item[a)] The second fundamental form $II$ of $Y$ in $N$ and all its covariant derivatives along $Y$ are bounded,
\item[b)] There exists $\varepsilon_Y>0$ such that the map $\phi_{\nu_Y}:Y\times(-\varepsilon_Y,\varepsilon_Y) \rightarrow N$ defined as $(x,z) \mapsto \phi_{\nu_Y}(x,z)\doteq exp_x(z \nu_{Y,x})$ is injective. 
\end{itemize}
\end{defn} 

These last two definitions can be combined to introduce the following notable class of Riemannian manifolds

\begin{defn}\label{Def: Riemannian manifold with boundary and of bounded geometry}
Let $(N,h)$ be a Riemannian manifold with $\partial N\neq\emptyset$. We say that $(N,h)$ is of bounded geometry if there exists a Riemannian manifold of bounded geometry $(N^\prime,h^\prime)$ of the same dimension as $N$ such that:
\begin{itemize}
\item[a)] $N \subset N^\prime$ and $h = h^\prime|_{N}$
\item[b)] $(\partial N, \iota^*h^\prime)$ is a bounded geometry submanifold of $N^\prime$, where $\iota:\partial N \rightarrow N^\prime$ is the embedding map.
\end{itemize}
\end{defn}

\begin{rem}
Observe that Definition \ref{Def: Riemannian manifold with boundary and of bounded geometry} is independent from the choice of $N^\prime$. For completeness, we stress that an equivalent definition which does not require introducing $N^\prime$ can be formulated, see for example \cite{Sch01}.
\end{rem}

In the following we shall introduce Sobolev spaces on a Riemannian manifold $(N,h)$ with boundary and of bounded geometry such that $\dim N=n$. In particular we shall recollect succinctly the main results of \cite[Sec. 2.4]{ANN16} to where we refer for further details. In the following, we denote with $r_{inj}(N)$ and $r_{inj}(\partial N)$, the injectivity radius of $N$, $\partial N$ respectively while $\delta>0$ is such that the normal exponential map $exp^\perp: \partial N\times[0,\delta) \rightarrow N$ is injective. With these data let
\begin{equation}\label{Eq: Fermi chart 1}
\begin{cases}
k_p : & B_r^{n-1}(0) \times[0,r) \rightarrow N \ \ \ \textit{if} \ \ p \in \partial N \\
 & (x,t)  \mapsto exp^\perp(exp_p^{\partial N}(x),t)\\
 k_p : & B_r^{n}(0) \rightarrow N \ \ \ \ \ \ \ \ \ \ \ \ \ \ \ \ \textit{if} \ \ p \in \mathring{N} \\
 & v \mapsto exp^N_p(v)\\
\end{cases},
\end{equation}
where we are implicitly identifying $T_p\partial N$ with $\mathbb{R}^{n-1}$, whenever $p\in\partial N$. In addition we introduce the sets
\begin{equation}\label{Eq: Fermi chart 2}
U_p(r) \doteq
\begin{cases}
k_p (B_r^{n-1}(0) \times [0,r)) \subset N\ \ \textit{if} \ \ p \in \partial N\\
k_p(B_r^{n}(0)) \ \ \ \ \ \ \ \ \ \ \ \ \ \ \ \ \ \ \ \ \ \ \textit{if} \ \ p \in \mathring{N} \\
\end{cases}
\end{equation}
where $r < \texttt{min} \left\{ \frac{1}{2} r_{inj}(N), \frac{1}{4}r_{inj}(\partial N), \frac{1}{2} r_\delta \right\}$. 

\begin{defn}\label{Def: Fermi coordinates}
Let $(N,h)$ be a Riemannian manifold with boundary and of bounded geometry of dimension $\dim N=n$. Let $$r < \texttt{min} \left\{ \frac{1}{2} r_{inj}(N), \frac{1}{4}r_{inj}(\partial N), \frac{\delta}{2} \right\}$$
For each $p\in\partial N$, we call Fermi coordinate chart the map $k_p:B^{n-1}_r(0) \times [0,r) \rightarrow W_p(r)$ with associated coordinates $(x,z):U_p(r) \rightarrow \mathbb{R}^{n-1} \times [0,	\infty) $.
\end{defn}

Observe that in view of Equations \eqref{Eq: Fermi chart 1} and \eqref{Eq: Fermi chart 2}, if $p\in\mathring{N}$, we can always consider geodesic neighbourhoods not intersecting $\partial N$ and endowed with normal coordinates. These data allow to introduce a distinguished covering

\begin{defn}\label{Def: covering of a Riemannian manfiold of bounded geoemtry}
Let $(N,h)$ be a Riemannian manifold with boundary and of bounded geometry. Let $0 < r < \texttt{min} \left\{ \frac{1}{2} r_{inj}(N), \frac{1}{4}r_{inj}(\partial N), \frac{\delta}{2} \right\}$. A subset $\{p_\gamma\}_{\gamma \in I}$, $i \subseteq \mathbb{N}$, is an {\em r-covering subset} of $N$ if:
\begin{itemize}
\item[a)] For each $R>0$, there exists $K_R \in \mathbb{N}$ such that, for each $p \in N$, the set $\{\gamma \in I \ | \ \textit{dist}(p_\gamma,p)<R\}$ has at most $K_R$ elements.
\item[b)] For each $\gamma \in I$, we have either $p_\gamma \in \partial N$ or $\textit{dist}(p_\gamma,\partial N) \geq r$.
\item[c)] $N \subset \bigcup_{\gamma \in I} U_{p_\gamma}(r)$, {\it cf.} Equation \eqref{Eq: Fermi chart 2}.
\end{itemize}
\end{defn}

\noindent At last, we need a partition of unity compatible with an r-covering set.

\begin{defn}\label{Def: partition of unity}
Under the same assumptions of Definition \ref{Def: covering of a Riemannian manfiold of bounded geoemtry}, a partition of unity $\{\phi_\gamma\}_{\gamma \in I}$ of $N$ is called an r-uniform partition of unity associated with the r-covering set $\{p_\gamma\}$ if:
\begin{itemize}
\item[a)] The support of each $\phi_\gamma$ is contained in $U_{p_\gamma}$, {\it cf.} Equation \eqref{Eq: Fermi chart 2},
\item[b)] For each multi-index $\alpha$, there exists $C_\alpha > 0$ such that $|\partial^\alpha \phi_\gamma|\leq C_\alpha$ for all $\gamma\in I$. Here the derivatives $\partial^\alpha$ are computed either in the normal geodesic or in the Fermi coordinates on $U_{p_\gamma}$ depending whether $p$ lies in $\mathring{N}=N\setminus\partial N$ or in $\partial N$.
\end{itemize}
\end{defn}

We have all ingredients to define Sobolev spaces on a Riemannian manifold $(N,h)$ with boundary and of bounded geometry. Let $\{\phi_\gamma\}$ be a uniform partition of unity associated with the r-covering set ${p_{\gamma}}$ as per Definition \ref{Def: partition of unity}. For every $k\in\mathbb{N}$ we call $k$-th Sobolev space, $H^k(N)$, the collection of all distributions $u \in \mathcal{D}^\prime(N)$ such that
\begin{equation}\label{Eq: Sobolev norm}
\| u \|^2_{H^k(N)} \doteq \sum_\gamma \|(\phi_{p_\gamma} u)\circ k_{p_\gamma} \|^2_{H^k}
\end{equation}
where $\|\cdot \|_{H^k}$ is the standard Sobolev space norm either on $\mathbb{R}^n$ or $\mathbb{R}^n_+$. 

It is important to stress that Equation \eqref{Eq: Sobolev norm} does not depend on the choice either of the $r$-covering and of the partition of unity . In addition we stress that, as in the case of a manifold without boundary \cite{GS13}, it turns out that $H^k(N)$ is equivalent to $W^{2,k}(N)$ which is the completion of 
$$\mathcal{E}^k(N)\doteq\{f\in C^\infty(N)\;|\;f,\nabla f,\dots,(\nabla)^k f\in L^2(N)\},$$
with respect to the norm 
$$\|f\|_{W^{2,k}(N)}=\left(\sum\limits_{i=0}^k\|(\nabla)^i f\|_{L^2(N)}\right)^{\frac{1}{2}},$$
where $\nabla$ is the covariant derivative built out of the Riemannian metric $h$, while $(\nabla)^i$ indicates the $i$-th covariant derivative. This notation is employed to disambiguate with $\nabla^i=h^{ij}\nabla_j$.

\vskip .2cm

To conclude the section we outline how the previous analysis can be extended to the case of Lorentzian manifolds. For simplicity we focus on the case without boundary, but the extension is straightforward. Following \cite{GOW17} we start from $(N,h)$ a Riemannian manifold of bounded geometry such that $\dim N=n$. In addition we call $BT^m_{m^\prime}(B_n(0,\frac{r_{inj}(N)}{2}),\delta_E)$, the space of all bounded tensors on the ball $B_n(0,\frac{r_{inj}(N)}{2})$ centered at the origin of the Euclidean space $(\mathbb{R}^n,\delta_E)$ where $\delta_E$ stands for the flat metric. For every $m,m^\prime\in\mathbb{N}\cup\{0\}$, we denote with $BT^m_{m^\prime}(N)$ the space of all rank $(m,m^\prime)$ tensors $T$ on $N$ such that, for any $p\in M$, calling $T_p\doteq(\exp_p\circ e_p)^*T$ where $e_p:(\mathbb{R}^n,\delta)\to (T_pN, h_p)$ is a linear isometry,  the family $\{T_p\}_{p\in M}$ is bounded on $BT^m_{m^\prime}(B_n(0,\frac{r_{inj}(N)}{2}),\delta_E)$. 

\begin{defn}\label{Def: Lorentzian manifold of bounded geometry}
	A smooth Lorentzian manifold $(M,g)$ is of bounded geometry if there exists a Riemannian metric $\widehat{g}$ on $M$ such that:
	\begin{itemize}
		\item[a)] $(M,\widehat{g})$ is of bounded geometry.
		\item[b)] $g \in BT^0_2 (M,\widehat{g})$ and $g^{-1} \in BT^2_0(M,\widehat{g})$.
	\end{itemize}
\end{defn}

\begin{rem}\label{Rem: Always bounded geometry}
	Henceforth we shall assume implicitly that {\em all} manifolds that we are considering are of bounded geometry. Although in many instances this property is not necessary, it becomes vital every time we need to invoke a partition of unity argument. 
\end{rem}

\section{Analytic Preliminaries}\label{Sec: Analytic Preliminaries}

In this section we introduce the basic analytic tools that we will need in the rest of the paper following mainly from \cite{GaWr18} and \cite{Vasy08}.

\subsection{b-pseudodifferential operators}\label{Sec: bPsiDos}

In this part of the section we introduce $b$-pseudodifferential operators and we stick to discussing the tools and the results that we need in the rest of the paper. We assume that the reader is already acquainted with the basic notions of b-calculus and, for further details we refer to the following introductory work, \cite{Gri00}.

In this section with $(M,g)$ we consider for definiteness a globally hyperbolic, asymptotically AdS spacetime with connected boundary $\partial M$ such that $\dim M=n$, {\it cf.} Definition \ref{Def: asymptotically AdS}. With $S^m(^bT^*M)$ we indicate  the set of symbols of order $m$ on $^bT^*M$, while with $\Psi^m_b(M)$ the properly supported $b$-pseudodifferential operators (b-$\Psi$DOs) of order $m$, $m\in\mathbb{R}$. Hence, calling $\dot{C}^\infty(M)$ ({\em resp.} $\dot{C}_0^\infty(M)$) the set of smooth ({\em resp.} smooth and compactly supported) functions in $M$, vanishing at the boundary with all derivatives, each $A\in\Psi^m_b(M)$ can be read as a continuous map $A:\dot{C}^\infty(M)\to\dot{C}^\infty(M)$ which can be extended to an endomorphism on $C^\infty(M)$. In addition, for any $A\in\Psi^m_b(M)$ there exists a principal symbol map 
\begin{equation}\label{Eq: principaly symbol map}
\sigma_{b,m} : A \mapsto a \in S^m(^bT^*M)/S^{m-1}(^bT^*M)
\end{equation}
which gives rise to an isomorphism 
$$\Psi_b^m(M)/\Psi_b^{m-1}(M) \cong S^m(^bT^*M)/S^{m-1}(^bT^*M)$$
This isomorphism and the definition of classical symbol over ${}^b T^*M$ yield as a consequence that $\Psi_b^m(M) \subset \Psi_b^n(M)$ if $m<n$. Notice in addition that the principal symbol of a b-$\Psi$DO is invariant under conjugation by a power of the boundary function $x$, {\it cf.} Definition \ref{Def: asymptotically AdS}. In other words if $A \in \Psi^m_b(M)$, then $x^{-s}A x^s \in \Psi^m_b(M)$ and $\sigma_{b,m}(x^{-s}Ax^s) =\sigma_{b,m}(A)$ for every $s \in \mathbb{R}$.

Since we are considering a Lorentzian manifold we can fix the metric induced volume density $\mu_g$ and, calling $A^*$ the formal adjoint of $A\in\Psi^m_b(M)$ with respect to the pairing induced by $\mu_g$, it turns out that $A^*\in\Psi^m_b(M)$. Furthermore $A^*$ admits the following asymptotic expansion \cite{McSa11}
\begin{equation}
\sigma(A^*)(z,k_z) \sim \sum_{\alpha} \frac{(-i)^{|\alpha|}}{\alpha !} \partial_{k_z}^\alpha \bigtriangledown_z^\alpha \overline{a(z,k_z)},\quad |k_z|\to\infty
\end{equation}  
where $\nabla_z$ denotes the covariant derivative induced from the metric $g$, acting on the point $z\in M$.
Existence of $A^*$ entails that $A$ extends also to an endomorphism of both $\mathcal{E}^\prime(M)$ and $\dot{\mathcal{E}}^\prime(M)$ the topological dual spaces of $C^\infty(M)$ and of $\dot{C}^\infty(M)$ respectively.

Given two pseudodifferential operators $A \in \Psi^m_b(M)$ and $B \in \Psi^n_b(M)$, the principal symbol of the composition $AB$ is $\sigma_{b,m+n}(AB) = \sigma_{b,m}(A)\cdot \sigma_{b,n}(B)$,  while their commutator $[A.B]\in \Psi^{m+n-1}_b(M)$ has a principal symbol which can be expressed locally in terms of Poisson brackets as \begin{gather*}
\sigma_{b,m+n-1}([A,B])=\{\sigma_{b,m}(A),\sigma_{b,n}(B) \}=\\
=\partial_\zeta\sigma_{b,m}(A)\cdot x\partial_x\sigma_{b,n}(B)-\partial_\zeta\sigma_{b,n}(B)\cdot  x\partial_x\sigma_{b,m}(A)+\\
+\sum\limits_{i=1}^{n-1}\left(\partial_{\eta_i}\sigma_{b,m}(A)\cdot \partial_{y_i}\sigma_{b,n}(B)-\partial_{\eta_i}\sigma_{b,n}(B)\cdot \partial_{y_i}\sigma_{b,m}(A)\right),
\end{gather*}
where $(x,y_i,\zeta,\eta_i)$, $i=1,\dots n-1$ are the local coordinates on an open subset of ${}^bT^*M$ introduced in Section \ref{Sec: Geometric preliminaries}.

In the following, we use bounded subsets of $\Psi_b^m(M)$ indexed by a real number in $(0,1)$. In order to make this notion precise, we equip $S^m({}^bT^*M)$ with the structure of a Fr\'echet space by means of the following family of seminorms
$$ \| a \|_{N} \ = \sup_{(z,k_z) \in K_i \times \mathbb{R}^n} \max_{|\alpha|+ |\gamma| \leq N} \dfrac{|\partial_z^\alpha \partial_\zeta^\gamma a(z,k_z) | }{\langle k_z \rangle^{m-|\gamma|}} $$
where $\langle k_z \rangle = (1+|k_z|^2)^{\frac{1}{2}}$, while $\{K_i\}_{i\in I}$, $I$ being an index set, is an exhaustion of $M$ by compact subsets. Hence one can endow $S^m({}^bT^*M)$ with a metric $d$ as follows: Given two symbols $a,b \in S^m({}^bT^*M)$, we call
\begin{equation*}
d(a,b) = \sum_{N \in \mathbb{N}} 2^{-N} \dfrac{\|a-b\|_{N}}{1+\|a-b \|_N}.
\end{equation*} 
Accordingly we say that a subset in $\Psi_b(M)$ is bounded if the subset of the symbols associated with the family of $\Psi$DOs is bounded.
\medskip

We are ready now to discuss the microlocal properties of b-pseudodifferential operators. We begin from the notion of elliptic b-$\Psi$DO.
\begin{defn}\label{Def: ellptic PsiDO}
A b-pseudodifferential operator $A \in \Psi^m_b(M)$ is {\em elliptic} at a point $q_0 \in \ {}^bT^*M \setminus\{0\}$ if there exists $c\in S^{-m}(^bT^*M)$ such that
\begin{equation*}
\sigma_{b,m}(A)\cdot c - 1 \in S^{-1}(^bT^*M)
\end{equation*} 
in a conic neighborhood of $q_0$. We call $ell_b(A)$ the (conic) subset of $^bT^*M \setminus\{0\}$ in which $A$ is elliptic.
\end{defn}

In the following we shall need the wavefront set both of a single and of a family of pseudodifferential operator. We recall here the definition, see \cite{Jos99} and, as far as notation is concerned, we adopt that of \cite{Hor1}:

\begin{defn}\label{Def: WF of PsiDO}
For any $P \in \Psi^m_b(M)$, we say that $(z_0,k_{z_0}) \notin WF^\prime_b(P)$ if the associated symbol $p(z,k_z)$ is such that, for every multi-indices $\gamma$ and for every $N\in\mathbb{N}$, there exists a constant $C_{N,\alpha,\gamma}$ such that
\begin{equation*}
|\partial_z^\alpha \partial^\gamma{k_z} p(z,k_z)| \leq C_{m,\alpha,\gamma} \langle k_z \rangle^{-N} 
\end{equation*}
for $z$ in a neighborhood of $z_0$ and $k_z$ in a conic neighborhood of $k_{z_0}$. 

Similarly, if $\mathcal{A}$ is a bounded subset of $\Psi_b^m(M)$ and $q \in {}^bT^*M$. We say that $q \not \in WF_b^\prime(\mathcal{A})$ if there exists $B \in \Psi_b(M)$, elliptic at $q$, such that $\{ BA : A \in \mathcal{A}\}$ is a bounded subset of $\Psi_b^{-\infty}(M)$.
\end{defn}

We recall a few notable consequences of Definition \ref{Def: WF of PsiDO}, see \cite{Jos99}. First of all $WF_b^\prime(P)=\emptyset$, if and only if $P \in \Psi^{-\infty}_b(M)$. In addition, given two bounded families of pseudodifferential operators, $\mathcal{A}$ and $\mathcal{B}$ it holds
\begin{equation*}
WF_b^\prime(\mathcal{A}+\mathcal{B}) \subset WF_b^\prime(\mathcal{A}) \cup WF_b^\prime(\mathcal{B})  \ \ \ \ WF_b^\prime(\mathcal{A}\mathcal{B}) \subset WF_b^\prime(\mathcal{A}) \cap WF_b^\prime(\mathcal{B})
\end{equation*}
Furthermore, if $B \in \Psi^m_b(M)$ is such that $WF_b^\prime(B) \cap WF_b^\prime(\mathcal{A}) = \emptyset$, then $\{AB : A \in \mathcal{A} \}$ is bounded in $\Psi^{-\infty}(M)$.

\begin{defn}\label{Def: microloal map}
Let $S  \subset \Psi^m_b(M)$ be a closed subspace. We say that a bounded linear map $M:S \rightarrow \Psi^k_b(M)$ is {\em microlocal} if $WF_b^\prime(M(A)) \subset WF_b^\prime(A)$ for all $A \in S$. 
\end{defn}

\noindent We can also microlocalize the notion of parametrix, see \cite{Vasy08}.

\begin{defn}[microlocal parametrix]\label{Def: microlocal paramterix}
Let $A \in \Psi^m_b(M)$ be elliptic as per Definition \ref{Def: ellptic PsiDO} in an open cone centered at a point $q \in {}^bT^*M \setminus\{0\}$. Then there exists a microlocal parametrix $\mathcal{G}$ for $A$ at $q$, namely $\mathcal{G}\in \Psi^{-m}_b(M)$ such that $\mathcal{G}A$ and $A\mathcal{G}$ are microlocally the identity operator near $q$. This means that $q \not \in WF_b^\prime(\mathcal{G}A-\mathbb{I})$ and $q \not \in WF_b^\prime(A\mathcal{G}-\mathbb{I})$. 
\end{defn}

Observe that, if $K \subset {}^bT^*M$ is a compact set and $A \in \Psi^m_b(M)$ is elliptic on $K$, then there exists $\mathcal{G}\in \Psi^{-m}_b(M)$ such that $K \cap WB_b^\prime (\mathcal{G}A-\mathbb{I}) = K \cap WB_b^\prime (A\mathcal{G}-\mathbb{I}) = \emptyset$. This entails that $E_1 =\mathcal{G}A-\mathbb{I}$ and $E_2 = A\mathcal{G}-\mathbb{I}$ lie in $\Psi^{-\infty}_b(K)$.

To conclude this part of the section, we stress that, in order to study the behavior of a b-pseudodifferential operator at the boundary, it is useful to introduce the notion of \textit{indicial family}, \cite{GaWr18}. Let $A \in \Psi_b^m(M)$. For a fixed boundary function $x$, {\it cf.} Definition \ref{Def: asymptotically AdS}, and for any $v \in \mathcal{C}^\infty(\partial M)$ we define the indicial family $\widehat{N}(A)(s):C^\infty(\partial M)\to C^\infty(\partial M)$ as:
\begin{equation}\label{Eq: indicial family}
\widehat{N}(A)(s)v = x^{-is}A \left( x^{is}u \right)|_{\partial M}
\end{equation}
where $u \in \mathcal{C}^\infty(M)$ is any function such that $u|_{\partial M}=v$. The indicial family does not depend on the choice of $u$ and it is manifestly an homomorphism of the algebra of pseudodifferential operators since, for all $A\in\Psi^m(M)$ and for all $B\in\Psi^{m^\prime}(M)$,
\begin{equation}\label{Eq: product of indicial families}
\widehat{N}(AB)(s) = \widehat{N}(A)(s) \circ \widehat{N}(B)(s).
\end{equation}

\subsection{Twisted Sobolev spaces}\label{subsec:twisted_sobolev}

Following the road paved in \cite{GaWr18}, a key ingredient of our analysis will be the Dirichlet form. To this end it is necessary to introduce a twisted version of the standard Sobolev spaces to account for the behaviour of the fields at the boundary, see also \cite{War1}. Furthermore, to deal with boundary conditions other than that of Dirichlet type, it is convenient to use twisted derivatives. 

Since we reckon that some readers might not find straightforward the necessity of twisting Sobolev spaces, we feel worth starting from a short motivational example. Let us consider the simplest case of a globally hyperbolic, asymptotically AdS spacetime as per Definition \ref{Def: asymptotically AdS}, namely $PAdS_2$, the Poincar\'e patch of the two dimensional anti de Sitter spacetime. $PAdS_2$ is a manifold diffeomorphic to $\mathbb{R} \times [0,\infty)$ whose metric is $g=x^{-2}\eta_2$ where $\eta$ is the two-dimensional Minkowski metric. Consider $\phi:PAdS_2\to\mathbb{R}$ obeying the Klein-Gordon equation
\begin{equation}\label{Eq: KG on PAdS2}
\Big( \square_g - m^2\Big) \phi =0 \Longrightarrow x^2\Big( \square_\eta - m^2 \Big) \phi = 0.
\end{equation}
For future convenience, we introduce the parameter $\nu = \frac{1}{2}\sqrt{1+4m^2}$ constrained to be positive. This is known is the theoretical physics literature as the Breitenlohner-Freedman bound \cite{BF82}. Equation \eqref{Eq: KG on PAdS2} can be solved by separation of variables using the ansatz $\phi(t,x)=F(t)H(x)$, where 
\begin{align*}
& F(t) = C_1 e^{\lambda t} + C_2 e^{-\lambda t}\\
& M(x) = C_3 \sqrt{x} J_\nu(-i\sqrt{\lambda}x)+C_4\sqrt{x}Y_\nu(-i\sqrt{\lambda}x)
\end{align*}
with $C_1,C_2,C_3,C_4 \in \mathbb{R}$ and $\lambda \in \mathbb{C}$ to be determined imposing boundary conditions at $x=0$ and at $x \rightarrow + \infty$. Here $J_\nu$ and $Y_\nu$ are the standard Bessel functions. In particular, $M_1(x) = \sqrt{x} J_\nu(-i\sqrt{\lambda}x)$ is the solution associated with the Dirichlet boundary condition at $x=0$, while $M_2(x) = \sqrt{x} Y_\nu(-i\sqrt{\lambda}x)$ is the Neumann counterpart. The indicial roots $\nu_\pm = \frac{1}{2} \pm \nu$ describe the behavior of the solutions near the boundary. In the following, we consider a relatively compact subset $U\subset PAdS_2$ such that $U \cap \partial PAdS_2 \neq \emptyset$ and we introduce the Dirichlet form 
\begin{equation}\label{Eq: Dirichlet form on PAdS2}
\mathcal{E}_D(\phi_1,\phi_2) = - \int_U g(d\phi_1,d\phi_2) d\mu_g = - \int_U \eta(d\phi_1,d\phi_2) dxdt,
\end{equation}
where $\phi_1,\phi_2$ are arbitrary solutions of Equation \eqref{Eq: KG on PAdS2}, while $g(\cdot,\cdot)$ is the metric induced pairing between $1$-forms. A direct inspection unveils that, if we choose as $\phi_1$ the solution of Equation \eqref{Eq: KG on PAdS2} with Dirichlet boundary conditions, $\eta(d\phi_1,d\phi_1) \sim x^{1+2\nu}$ close to the boundary $x=0$. Hence the $x$-integral in Equation \eqref{Eq: Dirichlet form on PAdS2} is always convergent, while, if we consider as $\phi_2$ the solution of Equation \eqref{Eq: KG on PAdS2} with Neumann boundary conditions, the $x$-integral is always divergent since $\eta(d\phi_2,d\phi_2) \sim x^{-1-2\nu}$. In order to bypass this hurdle, given a generic solution $\phi$ of Equation \eqref{Eq: KG on PAdS2}, we introduce the twisted derivatives
\begin{align}
& \widetilde{Q}_0 \phi = x^{\frac{1}{2}-\nu} \frac{\partial}{\partial t} \Big( x^{-\frac{1}{2}+\nu} \phi \Big) \label{Eq: Q0}\\
& \widetilde{Q}_1 \phi = x^{\frac{1}{2}-\nu} \frac{\partial}{\partial x} \Big( x^{-\frac{1}{2}+\nu} \phi \Big) \label{Eq: Q1}
\end{align}
Observe that this procedure only affects the derivative in the $x$ direction and that the power of the twisting factor is nothing but the indicial root $\nu_-$. In addition we define
\begin{equation}\label{Eq: twisted Dirichlet form on PAdS2} 
\mathcal{E}_0(\phi_1,\phi_2) = - \int_U g(d_{\widetilde{Q}}\phi_1,d_{\widetilde{Q}}\phi_2) d\mu_g = - \int_U \eta(d_{\widetilde{Q}}\phi_1,d_{\widetilde{Q}}\phi_2)dxdt
\end{equation}
where $d_{\widetilde{Q}}$ is the twisted differential defined as 
\begin{equation}\label{Eq: twsited differential PAdS2}
d_{\widetilde{Q}} \phi = x^{\nu_-} d \Big( x^{-\nu_-}\phi\Big).
\end{equation}
The integral along the $x$-direction in Equation \eqref{Eq: twisted Dirichlet form on PAdS2} is no longer divergent, both if we choose Dirichlet or Neumann boundary conditions.

Motivated by this example and following \cite{GaWr18}, we consider henceforth a generic globally hyperbolic, asymptotically AdS spacetime $(M,g)$ of dimension $\dim M=n$, {\it cf.} Definition \ref{Def: asymptotically AdS} and we introduce the space of twisted differential operators
\begin{equation*}
\textbf{Diff}^1_\nu(M) = \{ x^{\nu_-} D x^{- \nu_-} \ | \ D \in \textbf{Diff}^1(M) \}
\end{equation*}
where $\nu_- = \frac{n-1}{2}-\nu$, $\nu>0$ while $Diff^1(M)$ stands for the set of first order differential operators on $(M,g)$. In the following $\nu_-$ shall correspond to the lowest indicial root of the Klein-Gordon operator on $(M,g)$.

\begin{rem}
Since $\textbf{Diff}^1_\nu(M) \subset x^{-1}\textbf{Diff}^1_b(M)$ \cite[Lemma 3.1]{GaWr18}, it follows that $\textbf{Diff}^1_\nu(M)$ is finitely generated.
\end{rem}

\noindent Starting from these data we can introduce
\begin{equation}\label{Eq: rescaled L^2 space}
\mathcal{L}^2(M) \doteq L^2(M,x^2d\mu_g)
\end{equation}
and the corresponding twisted Sobolev space 

\begin{equation}\label{Eq: twisted Sobolev space}
\mathcal{H}^1(M) \doteq \Big\{ u \in \mathcal{L}^2(M) \ | \ Q u \in \mathcal{L}^2(M) \ \forall Q \in \textbf{Diff}^1_\nu(M) \Big\},
\end{equation}
whose norm is 
\begin{equation}\label{Eq: twisted norm}
\| u \|^2_{\mathcal{H}^1(M)} = \| u \|^2_{\mathcal{L}^2(M)} + \sum_{i=1}^n \|Q_i u\|^2_{\mathcal{L}^2(M)}
\end{equation}
where $\{Q_i\}_{i=1\dots n}$ is a generating set of $\textbf{Diff}^1_\nu(M)$.
In addition we shall be considering also $\mathcal{L}^2_{loc}(M)$ the space of locally square integrable functions over $M$ with respect to the measure $x^2d\mu_g$. Similarly one can introduce $\dot{\mathcal{L}}^2_{loc}(M)$ starting from $\dot{C}^\infty(M)$ in place of $C^\infty(M)$. 

All these spaces admit corresponding first order Sobolev spaces, which are indicated with  $\mathcal{H}^1_{loc}(M)$, $\dot{\mathcal{H}}^1_{loc}(M)$ respectively. Their topological duals are denoted instead with $\dot{\mathcal{H}}^{-1}_{loc}(M)$ and $\mathcal{H}^{-1}_{loc}(M)$. In addition we define 
\begin{equation}\label{Eq: H^1_0(m)}
\mathcal{H}^1_0(M)=\mathcal{H}^1_{loc}(M)\cap\mathcal{E}^\prime(M),
\end{equation}
where we denote with $\mathcal{E}^\prime(M)$ the topological dual space of $\dot{C}^\infty(M)$. Similarly one can define $\mathcal{H}^{-1}_0(M)$.

In the following we shall need two distinguished maps $\gamma_\pm : \mathcal{H}^1_{loc}(M)\rightarrow\mathcal{L}^2_{loc}(\partial M)$ playing the role of trace maps for twisted Sobolev spaces, hence generalizing to the case in hand the construction of \cite{GS13}. The first one can be individuated thanks to this result:

\begin{thm}[\cite{GaWr18}, Lemma 3.3]\label{lemma-exp1}
Let $\nu>0$, $2r=n-2$ and let $\mathbb{R}^n_+\doteq\mathbb{R}^{n-1}\times [0,\infty)$. If $u \in \mathcal{H}^1(\mathbb{R}^n_+)$, then the restriction of $u$ to $\mathbb{R}^{n-1}\times[0,\epsilon)$ for any $\epsilon>0$ admits an asymptotic expansion
\begin{equation}
u = x^{\nu_-}u_- + x^{r+1} H_b^1([0,\varepsilon);L^2(\mathbb{R}^{n-1}))
\end{equation}
where $u_- \in H^\nu(\mathbb{R}^{n-1})$ while $x$ is the coordinate along $[0,\infty)$. Furthermore, the application $u \mapsto \gamma_-u\doteq u_-$ is a continuous map from $\mathcal{H}^1(\mathbb{R}^n_+) \rightarrow H^\nu(\mathbb{R}^{n-1})$.
\end{thm}

In order to extend this result to a generic globally hyperbolic, asymptotically AdS spacetime, we can use a standard partition of unity argument to extend Theorem  \ref{lemma-exp1} to identify a continuous map
\begin{equation}\label{Eq: gamma-}
\gamma_-: \mathcal{H}^1_0(M)\to\mathcal{H}^\nu(\partial M)
\end{equation}
and similarly to $\mathcal{H}^1_{loc}(M)$. With a slight abuse of notation we shall employ the same symbol $\gamma_-$ as in Theorem \ref{lemma-exp1} since we reckon that no confusion can arise. Notice that Equation \eqref{Eq: gamma-} depends in general on the choice of the boundary function $x$, {\it cf.} Definition \ref{Def: asymptotically AdS}. Henceforth se shall assume that one such function has been selected and it will be kept unchanged throughout the paper.

For later purposes we give the following bound on the action of $\gamma_-$, see \cite{Gan15}.
\begin{lemma}\label{eq:gamma-bound}
Let $u \in \mathcal{H}_{loc}^1(M)$ be compactly supported. Then for any $\varepsilon>0$ there exists $C_\varepsilon>0$ such that $$\| \gamma_- u \|^2_{L^2(\partial M)} \leq \varepsilon \| u \|^2_{\mathcal{H}^1_{loc}(M)} +C_\varepsilon \| u \|^2_{\mathcal{L}^2_{loc}(M)}$$
\end{lemma}

To conclude this part, we need to introduce one last ingredient. At the beginning of the subsection, we have shown that it is possible to modify on $PAdS_2$ the Dirichlet form, see Equation \eqref{Eq: Dirichlet form on PAdS2} provided that one uses a twisted differential as in Equation \eqref{Eq: twsited differential PAdS2}. This idea can be generalized as follows

\begin{defn}\label{Def: twisting function}
We call {\em smooth twisting function} any $F \in x^{\nu_-}\mathcal{C}^\infty(M)$, such that $x^{-\nu_-}F>0$ is strictly positive on $M$.
\end{defn}

Hence, for any $B \in \textbf{Diff}^1(M)$, it holds that $FBF^{-1} \in \textbf{Diff}^1_\nu (M)$. Conversely, any $Q \in \textbf{Diff}^1_\nu (M)$ is always of the form $Q=FBF^{-1}$ for some $B \in \textbf{Diff}^1(M)$ while $F$ is a twisting function.

\subsection{Interaction with b-calculus and wavefront sets}\label{Sec: b-calculus and WF}
In this section we recall some useful results from \cite{GaWr18}, \cite{Vasy08} and \cite{Vasy10} concerning the interplay between properly supported b-$\Psi$DOs and $\textbf{Diff}_\nu(M)$. Throughout this section we are still assuming that $(M,g)$ is a globally hyperbolic, asymptotically AdS spacetime, {\it cf.} Definition \ref{Def: asymptotically AdS}. In view of Theorem \ref{Th: globally hyperbolic}, $M$ is isometric to $\mathbb{R}\times\Sigma$ and we can introduce a time coordinate $t$ running over the whole $\mathbb{R}$. For any twisting function $F$, {\it cf.} Definition \ref{Def: twisting function}, we call
$$Q_0=F\partial_x F^{-1}\in\textbf{Diff}^1_\nu (M).$$

\begin{lemma}[Lemma 3.7 of \cite{GaWr18}]\label{lem:diff-psi-interact}
Let $A \in \Psi^m_b(M)$ have compact support in $U \subset M$. There exist two pseudodifferential operators $A_1 \in \Psi^{m-1}_b(M)$ and $A_0 \in \Psi^m_b(M)$ such that $$[Q_0,A] = A_1 Q_0 + A_0 $$ where $\sigma_{b,m-1}(A_1)= -i\partial_\zeta \sigma_{b,m}(A)$ and $\sigma_{b,m}(A_0)= -i\partial_x \sigma_{b,m}(A)$, $\sigma_{b,m}$ being the principal symbol map as in Equation \eqref{Eq: principaly symbol map}, while $(x,\zeta)$ are the local coordinates on ${}^bT^*M$ introduced in Section \ref{Sec: Geometric preliminaries}. Also the maps $A \mapsto A_0$ and $A \mapsto A_1$ are microlocal in the sense of Definition \ref{Def: microloal map}. Furthermore,
$$ Q_0 A = A^\prime Q_0 + A^{\prime \prime} $$
for some $A^\prime, A^{\prime \prime} \in \Psi^m_b(M)$. The maps $A \mapsto A^\prime$ and $A \mapsto A^{\prime \prime}$ are microlocal.
\end{lemma}

In the following we will use bounded families of b-$\Psi DOs$ of fixed order, the most important case being that of a family of the form $\{A_r \in \Psi^m_b(M) \ | \ m <0, \ r \in (0,1)  \}$ bounded in $\Psi^0_b(M)$.

\begin{lemma}[Lemma 3.8 of \cite{GaWr18}, Lemma 3.2 of \cite{Vasy08}]\label{lemma:continuous-extension}
Let $A \in \Psi^0_b(M)$. Then $A$ is a continuous linear map
$$ \mathcal{H}^1_{loc/0}(M) \rightarrow \mathcal{H}^1_{loc/0}(M), \ \ \ \dot{\mathcal{H}}^1_{loc/0}(M) \rightarrow \dot{\mathcal{H}}^1_{loc/0}(M), $$
which extends per duality to a continuous map
$$ \dot{\mathcal{H}}^{-1}_{0/loc}(M) \rightarrow \dot{\mathcal{H}}^{-1}_{0/loc}(M), \ \ \ \mathcal{H}^{-1}_{0/loc}(M) \rightarrow \mathcal{H}^{-1}_{0/loc}(M).$$
\end{lemma}

\noindent As a direct consequence of this lemma, the following bound holds true.

\begin{prop}\label{prop-Psi0_Bound}
Let $A \in \Psi^0_b(M)$ have compact support in $U \subset M$. Then there exists $\chi \in C^\infty_0(U)$ such that 
\begin{equation*}
\| Au \|_{\mathcal{H}^k(M)} \leq C \| \chi u \|_{\mathcal{H}^k(M)}
\end{equation*}
for every $u \in \mathcal{H}^k_{loc}(M)$ with $k=\pm 1 $.
\end{prop}

A similar bound holds true if $u \in \dot{\mathcal{H}}^k_{loc}(M)$. We introduce, for $k=-1,0,1$, the subspaces of $\mathcal{H}^k(M)$ with additional regularity properties with respect to the action of b-pseudodifferential operators in $\Psi^m_b(M)$. These spaces allow us to get a better control on estimates like that of Proposition \ref{prop-Psi0_Bound}.

\begin{defn}\label{Def: tildeH spaces}
Let $k=-1,0,1$ and let $m \geq 0$. Given $u \in \mathcal{H}_{loc}^k(M)$, we say that $u \in \mathcal{H}_{loc}^{k,m}(M)$ if $Au \in \mathcal{H}_{loc}^k(M)$ for all $A \in \Psi^m_b(M)$. Furthermore, we define $\mathcal{H}^{k,\infty}(M)$ as:
\begin{equation}
\mathcal{H}^{k,\infty}(M) \doteq  \bigcap_{m=0}^\infty \mathcal{H}^{k,m}(M)  
\end{equation}
\end{defn}

\begin{rem}\label{Rem: Spazi Hk con m finito}
The spaces $\dot{\mathcal{H}}^{k,m}_{loc}(M)$, $\mathcal{H}^{k,m}(M)$ and  $\mathcal{H}^{k,m}_0(M)$ are defined in a similar way. Furthermore, as observed in \cite{Vasy08}, whenever $m$ is finite, it is enough to check that both $u$ and $Au$ lie in $\mathcal{H}^k_{loc}(M)$ for a single elliptic operator $A \in \Psi^m_b(M)$. As a consequence, for $u \in \mathcal{H}^{k,m}_{0}(M)$ with $m \geq 0$, we can define the following norm:
\begin{equation}
\|u\|_{\mathcal{H}^{k,m}(M)} = \|u\|_{\mathcal{H}^k(M)}+\|Au\|_{\mathcal{H}^k(M)}
\end{equation}
where $A$ is any elliptic b-pseudodifferential operator in $\Psi_b^m(M)$.
\end{rem}

\begin{defn}\label{Def: distinguished space of distributions}
Let $k = \pm 1$ and $m<0$. Let $A \in \Psi_b^{-m}(M)$ be a fixed pseudo-differential operator of positive order. We call $\mathcal{H}_{loc}^{k,m}(M)$ the set of the distributions $u \in \mathcal{D}^\prime(M)$ of the form
\begin{equation*}
u = u_1 + A u_2
\end{equation*}
where $u_1,u_2 \in \dot{\mathcal{H}}^k_{loc}(M)$.
\end{defn}

\begin{rem}
In the same spirit of Remark \ref{Rem: Spazi Hk con m finito}, we can define $\dot{\mathcal{H}}^{k,m}_{loc}(M)$ and $\mathcal{H}^{k,m}(M)$ in a similar way. Furthermore, when $m<0$ is finite, it is enough to check that both $u$ and $Au$ lie in $\mathcal{H}^k_{loc}(M)$ for a single elliptic operator $A \in \Psi^{-m}_b(M)$. 
\end{rem}

A notable consequence of these definitions can be summarized in the following lemma, whose proof can be found in \cite[Rem. 3.16]{Vasy08}.

\begin{lemma}\label{Lem: interaction between Hk and gamma-}
	Let $m<0$ and let $\mathcal{H}^{k,m}(M)$ be as in Definition \ref{Def: distinguished space of distributions}. Then $\gamma_-$ as per Equation \eqref{Eq: gamma-} extends to a continuous map
	$$\gamma_-:\mathcal{H}^{1,m}_{loc}(M)\to\mathcal{H}^{\nu+m}_{loc}(\partial M).$$
\end{lemma}

In the following we give the definition of wavefront set for $\mathcal{H}^{k,m}_{loc}(M)$, the counterpart for all other spaces following suit.

\begin{defn}\label{Def: wavefrontset for Hkm-loc}
Let $k=0,\pm 1$ and let $u \in \mathcal{H}^{k,m}_{loc}(M)$, $m \in \mathbb{R}$. Given $q \in {}^bT^*M \setminus\{0\} $, we say that $q \not \in WF_b^{k,m}(u)$ if there exists $A \in \Psi_b^m(M)$ such that $q \in ell_b(A)$ and $Au \in \mathcal{H}^k_{loc}(M)$, where $ell_b$ stands for the elliptic set as per Definition \ref{Def: ellptic PsiDO}. When $m=+\infty$, we say that $q \not \in WF_b^{k,\infty}(M)$ if there exists $A \in \Psi^0_b(M)$ such that $q \in ell_b(A)$ and $Au \in \mathcal{H}^{k,\infty}_{loc}(M)$.
\end{defn}

\noindent Definition \ref{Def: wavefrontset for Hkm-loc} is microlocal in the following sense:
\begin{equation*}
WF_b^{k,m}(Au) \subset WF_b^{k,m-s}(u) \cup WF_b^\prime(A)
\end{equation*}
for each $A \in \Psi_b^s(M)$, $s\geq 0$. Yet, sometimes, it is useful to have at our disposal a more refined bound. Combining the results in \cite{GaWr18} and \cite{Vasy08} the following lemma descends.

\begin{lemma}\label{prop-boundH1}
Let $\mathcal{A}$ be a bounded family in $\Psi^s_b(M)$ and let $G \in \Psi^s_b(M)$ be such that $WF_b^\prime (\mathcal{A}) \subset ell_b(G)$. Suppose that $\mathcal{A}$ and $G$ have compact support in $U \subset M$. Let $m \in \mathbb{R}$ and $k = \pm 1$. Then there exist $\chi \in \mathcal{C}_0^\infty(U)$ and a constant $C>0$ such that
\begin{equation*}
\| Au \|_{\mathcal{H}^k(M)} \leq C \left( \| Gu \|_{\mathcal{H}^k(M)}+\| \chi u \|_{\mathcal{H}^{k,m}(M)}  \right)
\end{equation*}
for every $u \in \mathcal{H}^{k,m}_{loc}(M)$ with $WF_b^{k,s}(u) \cap WF^\prime_b(G) = \emptyset$ and for every $A \in \mathcal{A}$.
\end{lemma}

\noindent A notable consequence is

\begin{lemma}[Lemma 3.13 \cite{GaWr18}]\label{prop-boundL2}
Let $\mathcal{A}$ be a bounded family of pseudodifferential operators in $\Psi^s_b(M)$ and let $G \in \Psi^{s-1}_b(M)$ be such that $WF_b^\prime (\mathcal{A}) \subset ell_b(G)$. Suppose that $\mathcal{A}$ and $G$ have compact support in $U \subset M$. Let $m \in \mathbb{R}$ and let $k = \pm 1$. Then there exist $\chi \in \mathcal{C}_0^\infty(U)$ and a constant $C>0$ such that
\begin{equation*}
\| Au \|_{\mathcal{L}^2(M)} \leq C \left( \| Gu \|_{\mathcal{H}^k(M)}+\| \chi u \|_{\mathcal{H}^{k,m}(M)}  \right)
\end{equation*}
for every $u \in \mathcal{H}^{k,m}_{loc}(M)$ with $WF_b^{k,s-1}(u) \cap WF^\prime_b(G) = \emptyset$ and for every $A \in \mathcal{A}$.
\end{lemma}
These two lemmas play a pivotal role in the following, when we employ energy estimates to prove the propagation of singularity theorem.

\subsection{Asymptotic expansion and traces}\label{Sec: Asymptotic expansion}
In Section \ref{subsec:twisted_sobolev}, we have already individuated a trace $\gamma_-$ in Equation \eqref{Eq: gamma-}. Here we tackle the problem of finding a second one $\gamma_+$. To this end we consider once more a globally hyperbolic, asymptotically AdS spacetime $(M,g)$ and the associated Klein-Gordon operator acting on scalar functions $\phi:M\to\mathbb{R}$
\begin{equation}\label{Eq: KG on M}
P\phi=\left(\Box_g-m^2\right)\phi=\left(x^2\Box_{\widetilde{g}}-m^2\right)\phi=0,
\end{equation}
where $\Box_g$ is the D'Alembert wave operator built out of the metric, while $m^2\geq 0$ plays the role of the squared mass. Here $\widetilde{g}$ is the metric whose associated line element is $-dx^2+h_x$, see Equation \eqref{Eq: metric near the boundary}.

Following \cite{GaWr18}, we can now introduce a family of functional spaces enjoying additional regularity with respect to the Klein-Gordon operator.

\begin{defn}\label{Def: chik spaces}
Let $(M,g)$ be a globally hyperbolic, asymptotically anti-de Sitter spacetime and let $P$ be the Klein-Gordon operator as in Equation \eqref{Eq: KG on M}. For all $m \in \mathbb{R}$, we define the Frech\'et spaces 
\begin{equation}\label{Eq: spazi chik}
\mathcal{X}^m(M) = \{u \in \mathcal{H}^{1,m}_{loc}(M)\; |\; Pu \in x^2 \mathcal{H}^{0,m}_{loc}(M) \},
\end{equation}
with respect to the seminorms
\begin{equation}\label{Eq: seminorme chik}
\norm{u}_{\mathcal{X}^m(M)} = \norm{\phi u}_{\mathcal{H}^{1,m}(M)}+\norm{x^{-2}\phi P u}_{\mathcal{H}^{0,m}(M)},
\end{equation}
where $\phi$ is a suitable smooth and compactly supported function.
\end{defn}

\begin{rem}\label{Rem: chik spaces for relatively compact}
If $K$ is any relatively compact subset of $M$ we can introduce in analogy to Definition \ref{Def: tildeH spaces} the space $\mathcal{H}^{k,m}(K)$, $k=0,1$, $m>0$, as well as 
$$\mathcal{X}^m(K) = \{u \in \mathcal{H}^{1,m}(K)\; |\; x^{-2}Pu \in \mathcal{H}^{0,m}(K) \},$$
endowed with the norm
\begin{equation}\label{Eq: chikV norms}
\norm{u}_{\mathcal{X}^m(K)} = \norm{ u }_{\mathcal{H}^{1,m}(K)}+\norm{x^{-2} P u}_{\mathcal{H}^{0,m}(K)}
\end{equation}	
\end{rem}

\noindent In the following we show that, starting from $\chi^k_{loc}(M)$, it is possible to improve the expansion given in Theorem \ref{lemma-exp1}. 

\begin{lemma}[Lemma 4.6 in \cite{GaWr18}]\label{lemma-exp2}
Let $(\mathbb{R}^n_+,g)$ be an asymptotically AdS spacetime as per Definition \ref{Def: asymptotically AdS} such that, with respect to the standard Cartesian coordinates, the line elements reads
\begin{equation}
g = \frac{-dx^2+h_{ab}dy^ady^b}{x^2}.
\end{equation}
Consider an admissible twisting function $F$, as per Definition \ref{Def: twisting function} such that at $x=0$ $x^{-\nu_-} F = 1$, where $\nu_-=\frac{1}{2}-\nu$ is the indicial root. If $u \in \mathcal{H}^{1,k}_0(\mathbb{R}^n_+)$ and $Pu \in x^2 \mathcal{H}^{0,k}_0(\mathbb{R}^n_+)$ for $k \geq 0$, then, for any $\epsilon>0$ the restriction of $u$ to $\mathbb{R}^{n-1}\times [0,\epsilon)$ admits an asymptotic expansion
\begin{equation}\label{Eq: asymptotic expansion in chik}
u = Fu_- + x^{\nu_+}u_+ + x^{r+2}H_b^{k+2}([0,\varepsilon);H^{k-3}(\mathbb{R}^{n-1}))
\end{equation}
where $2r=n-2$, $u_- \in H^{\nu+k}(\mathbb{R}^{n-1})$ and $u_+ \in H^{-1-2\nu+k}(\mathbb{R}^{n-1})$.
\end{lemma}
\noindent This result together with Theorem \ref{lemma-exp1} allows us to define the sought trace $\gamma_+$ on $\chi^\infty(M)$

\begin{equation}\label{Eq: gamma+}
\gamma_+ u = x^{1-2\nu}\partial_x(F^{-1}u)|_{\partial X}
\end{equation}
Observe that, working in a special coordinate patch, the restriction of $u$ to the boundary can be written as
\begin{equation}
u = x^{\nu_-}F u_- + x^{\nu_+}u_+ + u_2,  \ \ \ \ u_2 \in x^2\mathcal{H}^{2,\infty}_{loc}([0,\varepsilon)\times \mathbb{R}^{n-1}).
\end{equation} 
In these coordinates $\gamma_+ u = 2\nu u_+$. 

\begin{rem}
The second term of the expansion, of the form $x^{\nu_+}u_+$ is the leading term of the asymptotic behavior of a solution of the Klein-Gordon equation with Dirichlet boundary conditions on an anti-de Sitter spacetime. For this reason, we refer to $\gamma_+$ as the Dirichlet trace.
\end{rem}

\subsection{The twisted Dirichlet energy form}\label{Sec: Twisted Dirichlet form}

The last part of this chapter is devoted to the construction of a twisted Dirichlet form for the case in hand. Recall that we are considering a Klein-Gordon operator as per Equation \eqref{Eq: KG on M} on a globally hyperbolic, asymptotically AdS spacetime, {\it cf.} Definition \ref{Def: asymptotically AdS}.

Consider in addition a twisting function $F$ as per Definition \ref{Def: twisting function} and, in analogy to Equation \eqref{Eq: twsited differential PAdS2}, define the twisted differential 
$$d_F \doteq F \circ d \circ F^{-1},$$
whose action on smooth functions vanishing at $\partial M$ together with all its derivatives is
\begin{equation*}
d_F:\dot{C}^\infty(M) \rightarrow \dot{C}^\infty(M;T^*M),\quad v\mapsto d_F v = Fd(F^{-1}v) = dv + v F^{-1} dF
\end{equation*}

With these data, for every $u,v \in \mathcal{L}^2_{loc}(M)$, we define the twisted Dirichlet form by:
\begin{equation}\label{Eq: twisted Dirichlet form}
\mathcal{E}_0(u,v) = - \int\limits_M g(d_Fu,d_F \overline{v}) d\mu_g,
\end{equation}
where $d\mu_g$ is the metric induced volume form. 

Observe that, if $u,v \in \mathcal{H}^1_{loc}(M)$ and the intersection of their supports is compact, then $\mathcal{E}_0(u,v)$ is finite.  

\begin{rem}\label{Rem: Klein-Gordon operator and twisted differentials}
We can conveniently express the Klein-Gordon operator in Equation \eqref{Eq: KG on M} in terms of the twisted differentials:
\begin{equation}
P = -(d_F)^\dagger d_F + F^{-1}P(F),
\end{equation}
where $(d_F)^\dagger$ is the formal adjoint of $d_F$ with respect to inner product on $L^2(M;d\mu_g)$. As observed in \cite{War1}, twisted differentials can be used to regularize the energy form in the case in which the multiplication by $S_F = F^{-1}P(F) \in \mathcal{C}^\infty(\mathring{M})$ is a bounded operator from $\mathcal{L}^2(M)$ to $x^2\mathcal{L}^2(M)$, where $x$ is the boundary function.	
\end{rem}

In view of this remark and as in \cite{GaWr18}, we consider only a subset of the twisting functions:

\begin{defn}\label{Def: admissible twisting function}
	A twisting function $F$ as in Definition \ref{Def: twisting function} is called {\em admissible} if $S_F\doteq F^{-1}PF \in x^2 L^\infty(M)$ where $P$ is the Klein-Gordon operator as in Equation \eqref{Eq: KG on M}.
\end{defn}

A first notable application of admissible twisting functions is the following: 
Let $\nu \in (0,1)$ and suppose that $u \in \chi^\infty_{loc}(M)$, {\it cf.} Definition \ref{Def: chik spaces} while $v \in \mathcal{H}^1_{0}(M)$. Then the following Green's formula holds true:

\begin{equation}\label{eq:GreenFormula}
\int Pu \cdot \overline{v} \ d\mu_g = \mathcal{E}_0(u,v) + \int S_Fu \cdot \bar{v} \ d\mu_g + \int \gamma_+u \cdot \gamma_-\bar{v} \ d\mu_h,
\end{equation}
where $d\mu_h$ is the volume form induced by $h$, the pull-back of $g$ to $\partial M$. Actually it is possible to extend both the realm of applicability of Equation \eqref{eq:GreenFormula} and of the trace map $\gamma_+$ as in Equation \eqref{Eq: gamma+} as discussed in \cite[Lemma 4.8]{GaWr18}:

\begin{lemma}\label{th:GreenLemma}
	The map $\gamma_+$ as per Equation \eqref{Eq: gamma+} can be extended to a bounded map 
	$$\gamma_+:\mathcal{X}^k(M)\to\mathcal{H}^{k-\nu}_{loc}(\partial M),\quad\forall k\in\mathbb{R}$$
	and, if $u \in \mathcal{X}^k(M)$, the Green's formula \eqref{eq:GreenFormula} holds true for every $v \in \mathcal{H}^{1,-k}_{0}(M)$.
\end{lemma}

\section{Boundary value problem and $b-\Psi$DOs}\label{Sec: Boundary Value Problem}

In this section we use the geometric and analytic tools introduced in Section \ref{Sec: Geometric preliminaries} and \ref{Sec: Analytic Preliminaries} to introduce a distinguished class of boundary conditions for the Klein-Gordon equation, see Equation \eqref{Eq: KG on M}. Recall once more that $(M,g)$ refers to a globally hyperbolic, asymptotically AdS spacetime of dimension $n=\dim M$. Using the same nomenclature as in Subsection \ref{subsec:twisted_sobolev}, we introduce 
\begin{equation}\label{Eq: nu parameter}
\nu=\frac{1}{2}\sqrt{1+4m^2}\quad\textrm{and}\quad\nu_\pm=\frac{1}{2}\pm\nu,
\end{equation}
where $\nu_\pm$ are the {\em indicial roots}. Henceforth we shall consider only the case $\nu\in(0,1)$. The case $\nu=0$ could be included but it would require each time a separate analysis. For clarity purposes, we avoid considering such extremal scenario. The values of the mass for which $\nu\geq 1$ are not a priori pathological, but they are known not to require a boundary condition, see \cite{Dappiaggi:2017wvj, DDF18}.  

\subsection{Boundary conditions and the associated Dirichlet form}\label{Sec: Boundary conditions and the associated Dirichlet form}

In this section, we formulate the dynamical problem, we are interested in, so that the boundary condition is implemented by suitable $\Theta\in\Psi^k_b(\partial M)$. Formally, we look for $u\in H^1_{loc}(M)$ such that 
$$Pu=0,\quad\textrm{and}\quad \gamma_+u = \Theta\gamma_-u,$$
where $P$ is the Klein-Gordon operator as in Equation \eqref{Eq: KG on M}. Observe that, in order for this problem to be defined in a strong sense, we also need that $Pu\in x^2L^2_{loc}(M)$. Rather than focusing on this issue we give a weak formulation. More precisely consider $\Theta\in\Psi^k_b(\partial M)$ and define
\begin{equation}\label{Eq: Energy form with Theta}
\mathcal{E}_\Theta(u,v) = \mathcal{E}_0(u,v) + \int_M S_Fu \cdot \overline{v} d\mu_g + \int\limits_{\partial M} \Theta\gamma_- u \cdot \gamma_- \overline{v} d\mu_h,
\end{equation}
where $u \in \mathcal{H}^1_{loc}(M)$, while $v \in {\mathcal{H}}^1_{0}(M)$ whereas $F$ is an admissible twisting function, {\it cf.} Definition \ref{Def: admissible twisting function}, whose existence is assumed. Hence, we can introduce $P_\Theta : \mathcal{H}^1_{loc}(M) \rightarrow \dot{\mathcal{H}}^{-1}_{loc}(M)$ by
\begin{equation}\label{Eq: P_Theta}
\langle P_\Theta u, v \rangle = \mathcal{E}_\Theta(u,v)
\end{equation}
Observe that, on account of the regularity of $\gamma_-u$, we can extend $P_\Theta$ as an operator $P_\Theta: \mathcal{H}^{1,m}_{loc}(M) \rightarrow \dot{\mathcal{H}}^{-1,m}_{loc}(M)$, $m\in\mathbb{R}$ \cite{GaWr18}.

\begin{rem}\label{Rem: Pairing}
	Observe that, in this work, different pairings appear. For the sake of the simplicity of the notation we shall always use the symbol $\langle,\rangle$ since we reckon that the exact meaning can be understood from the context without risk of confusion. For example, in Equation \eqref{Eq: P_Theta}, the brackets $\langle,\rangle$ indicate the pairing between $\mathcal{H}^1(M)$ and $\dot{\mathcal{H}}^{-1}(M)$. 
\end{rem}

We report now a few microlocal estimates for the Dirichlet form, the first being the following upper bound.

\begin{lemma}[\cite{GaWr18}, Lemma 5.2]\label{lemma:DirichletE0Bound}
Let $U \subset M$ be a coordinate patch such that $U\cap\partial M\neq\emptyset$ and let $m \leq 0$. Let $\mathcal{A}=\{ A_r \ | \ r \in (0,1) \}$ be a bounded subset of $\Psi^s_b(M)$, $s\in\mathbb{R}$ with compact support in $U$, such that
\begin{equation*}
A_r \in \Psi^m_b(M) \ \textit{for each} \ r \in (0,1)
\end{equation*}
Let $G_1 \in \Psi^{s-1/2}_b(M)$ be elliptic on $WF^\prime_b(\mathcal{A})\subset{}^bT^*M\setminus\{0\}$, with compact support in $U$, {\it cf.} Definition \ref{Def: ellptic PsiDO}. Then there exist $C_0 > 0$ and $\chi \in \mathcal{C}^\infty_0(U)$ such that
\begin{equation*}
\mathcal{E}_0(A_r u,A_r u) \leq \mathcal{E}_0(u,A_r^* A_r u) + C_0 \left(\|G_1 u \|^2_{\mathcal{H}^1(M)}+ \|\chi u \|^2_{\mathcal{H}^{1,m}(M)} \right)
\end{equation*}
for every $r \in (0,1)$ and every $u \in \mathcal{H}^{1,m}(M)$, provided that $$WF^{1,s-1/2}_b(u) \cap WF^\prime_b(G_1) = \emptyset,$$
where $WF^{1,s-1/2}_b(u)$ is built as in Definition \ref{Def: wavefrontset for Hkm-loc}.
\end{lemma}

At this point we can prove an estimate for the boundary value problem associated with the pseudodifferential operator $\Theta\in\Psi^k_b(\partial M)$. As mentioned in the introduction we can control two different classes 
\begin{itemize}
	\item $\Theta\in\Psi^k_b(\partial M)$ with $k\leq 0$,
	\item $\Theta\in\Psi^k_b(\partial M)$ with $0<k\leq 2$.
\end{itemize}

In particular, we bound the difference between a generic positive-definite sesquilinear pairing form $\mathcal{Q}$ and $\mathcal{E}_0$ as per Equation \eqref{Eq: twisted Dirichlet form}. We start from $k \leq 0$.

\begin{lemma}\label{Lem: lemma-bound-quad-form-v1}
Let $U \subset M$ be a coordinate patch such that $U\cap\partial M\neq\emptyset$ and let $m \leq 0$. Suppose $\Theta \in \Psi_b^k(\partial M)$ with $k\leq 0$. Let $\mathcal{A}=\{ A_r \ | \ r \in (0,1) \}$ be a bounded subset of $\Psi^s_b(M)$, $s \in \mathbb{R}$, with compact support in $U$, such that
\begin{equation*}
A_r \in \Psi^m_b(M) \ \textit{for each} \ r \in (0,1)
\end{equation*}
Let $G_0 \in \Psi^s_b(M)$ be elliptic on $WF^\prime_b(\mathcal{A})$ and $G_1 \in \Psi_b^{s-1/2}(M)$ be elliptic on $WF^\prime_b(\mathcal{A})$, both with compact support in $U$, {\it cf.} Definition \ref{Def: ellptic PsiDO}. In addition let $\mathcal{E}_0$ and $\mathcal{Q}$ be respectively the twisted Dirichlet form on $\mathcal{L}^2_{loc}(M)$ and a generic positive-definite sesquilinear pairing. Then there exists $C_0 > 0$ and $\chi \in \mathcal{C}^\infty_0(U)$ such that
\begin{equation*}
\begin{split}
& \mathcal{E}_0(A_r u,A_r u) - \varepsilon \mathcal{Q}(A_r u,A_r u) \leq \\ & C_0 \Big( \|\chi u \|^2_{\mathcal{H}^{1,m}(M)} + \|\chi P_\Theta u \|^2_{\dot{\mathcal{H}}^{-1,m}(M)} + \| G_0 P_\Theta u\|^2_{\dot{\mathcal{H}}^{-1}(M)}  +  \|G_1 u \|^2_{\mathcal{H}^1(M)} \Big)
\end{split}
\end{equation*}
for every $r \in (0,1)$ and every $u \in \mathcal{H}^{1,m}(M)$, provided that the following conditions are met:
$$\ WF^{-1,s}_b(P_\Theta u) \cap WF^\prime_b(G_0) = \emptyset$$
$$ WF_b^{1,s-1/2}(u) \cap WF_b^\prime(G_1) = \emptyset $$
\end{lemma}

\begin{proof}
We start by emphasizing that, each $\Theta\in\Psi^k(\partial M)$ can be trivially extended to the whole $M$ by considering it as independent from the coordinate $x$. With a slight abuse of notation, we shall use the symbol $\Theta$ in both cases.

In order to bound 
\begin{equation}\label{eq:formbound}
\mathcal{E}_0(A_r u, A_r u)-\varepsilon \mathcal{Q}(A_r u, A_r u),
\end{equation}
it is convenient to rewrite this expression as
\begin{equation} \tag{\ref{eq:formbound}}
\begin{split}
\mathcal{E}_0(A_r u, A_r u)  - \mathcal{E}_0(u,A_r^*A_r u)  +\\
+ \mathcal{E}_0(u,A_r^*A_r u) - \mathcal{E}_\Theta(u,A_r^* A_r u)+\\
+ \mathcal{E}_\Theta(u,A_r^* A_r u) -\varepsilon \mathcal{Q}(A_r u, A_r u)
\end{split}
\end{equation}
Applying Lemma \ref{lemma:DirichletE0Bound}, we can bound the first line of Equation \eqref{eq:formbound} as
$$ \mathcal{E}_0(A_r u,A_r u) - \mathcal{E}_0(u,A_r^* A_r u) \leq C_0 \left(\|G_1 u \|^2_{\mathcal{H}^1(M)}+ \|\chi u \|^2_{\mathcal{H}^{1,m}(M)} \right)$$
The third line can be controlled as follows: Calling $f = P_\Theta u$, we can write $\mathcal{E}_\Theta(u,A^*_r A_r u) = \langle A_r f, A_r u	\rangle$. Using the pairing between $\mathcal{H}^1(M)$ and $\dot{\mathcal{H}}^{-1}(M)$,
$$\langle A_r f,A_r u\rangle \leq \|A_r f \|_{\dot{\mathcal{H}}^{-1}(M)}  \| A_r u \|_{\mathcal{H}^1(M)}.$$
Since for $C \geq 1/2$, $ab \leq C(a^2+b^2)$ for any $a,b \in \mathbb{R}$, then it holds that 
\begin{equation} \label{eq:energy_estimate1}
\|A_r f \|_{\dot{\mathcal{H}}^{-1}(M)}  \| A_r u \|_{\mathcal{H}^1(M)} \leq C \left( \|A_r f \|_{\dot{\mathcal{H}}^{-1}(M)}^2 +  \| A_r u \|_{H^1(M)}^2 \right)
\end{equation}

Using Lemma \ref{prop-boundL2} one obtains
\begin{gather*} 
\| A_r u \|_{\mathcal{H}^1(M)}^2 - \varepsilon \mathcal{Q}(A_r u, A_r u)  \leq \\
\leq C_1 \left( \| G_0 u \|_{\mathcal{H}^1(M)} + \| \chi u \|_{\mathcal{H}^{1,m}(M)} \right)^2 - \varepsilon \mathcal{Q}(A_r u, A_r u) 
\end{gather*}
where $G_0 \in \Psi^{s}_b(M)$. Applying again the inequality $ab \leq C(a^2+b^2)$, the second term in Equation \eqref{eq:energy_estimate1} is bounded by
\begin{equation}\label{eq-est1}
\| A_r u \|_{H^1(M)}^2 \leq C \left(  \| G_0 u \|_{\mathcal{H}^1(M)}^2 + \| \chi u \|_{\mathcal{H}^{1,m}(M)}^2 \right) - \varepsilon \mathcal{Q}(A_r u, A_r u),
\end{equation}
where  $G_0 \in \Psi^s_b(M)$.
We estimate the first term of \eqref{eq:energy_estimate1} using an analogue procedure, this time with the help of Lemma \ref{prop-boundH1}:
\begin{equation}\label{eq-est2}
\|A_r f \|_{\dot{\mathcal{H}}^{-1}(M)}  \leq C_2 \left( \|\chi f \|_{\dot{\mathcal{H}}^{-1}(M)} + \| G_0 u\|_{\dot{\mathcal{H}}^{-1,m}(M)} \right),
\end{equation}
where $G_0 \in \Psi^s_b(M)$.
Combining Equations \eqref{eq-est1} and \eqref{eq-est2}, we obtain the bound
\begin{equation}
\begin{split}
& |\langle A_r f,A_r u \rangle| \leq \varepsilon \mathcal{Q}(A_r u, A_r u) +\\
&  +  C \left(  \| G_0 u \|_{\mathcal{H}^1(M)}^2 + \| \chi u \|_{\mathcal{H}^{1,m}(M)}^2 + \|\chi f \|_{\dot{\mathcal{H}}^{-1}(M)} + \| G_0 u \|_{\dot{\mathcal{H}}^{-1,m}(M)} \right)
\end{split}
\end{equation}

\noindent At last, we control the second line in Equation \eqref{eq:formbound}.
\begin{eqnarray}\label{eq:ThetaBound}
\mathcal{E}_\Theta(u,A_r^* A_r u) - \mathcal{E}_0(u,A_r^* A_r u) = \notag\\
= \langle x^{-2}S_Fu, A_r^* A_r u \rangle + \langle \Theta \gamma_- u, \gamma_-(A_r^* A_r u) \rangle_{\partial M}
\end{eqnarray}
Using that $S_F \in x^2 \mathcal{C}^\infty(M)$, {\it cf.} Definition \ref{Def: admissible twisting function}, it holds
\begin{eqnarray*}
\langle x^{-2}S_Fu,A^*_r A_r u \rangle  =   \int_U \  x^{-2}S_Fu \overline{A^*_r A_r u} \ x^2 d\mu_g  \leq \\
\leq  \max_{x \in \pi_1 \circ \textrm{supp} (A_r)} \Big\lvert x^{-2}S_F \Big\rvert \cdot |\langle u,A^*_r A_r u \rangle |.
\end{eqnarray*}
In order to control $\langle u,A^*_r A_r u \rangle  = \| A_r u \|^2_{\mathcal{L}^2(M)}$, we use the same algebraic trick as above. On account of  Lemma \ref{prop-boundL2}, it holds
\begin{equation*}
|\langle x^{-2}S_Fu,A^*_r A_r u \rangle | \leq C_0 \left( \| \chi u \|^2_{\mathcal{H}^{1,m}(M)} + \| G_1 u \|^2_{\mathcal{H}^1(M)} \right).
\end{equation*}
At last, we focus on the boundary term $\langle \Theta\gamma_-u,\gamma_-(A^*_r A_r u) \rangle|_{\partial M}$. We recall that for every $B \in \Psi^m_b(M)$, it holds 
$$\gamma_-(Bu) = \left.\left(x^{-\nu_-} Bu \right)\right|_{\partial M} = \widehat{N}(B)(-i \nu_-)(\gamma_- u),$$
where $\widehat{N}(B)$ is the indicial family as per Equation \eqref{Eq: indicial family}. 
Extending $\Theta$ as explained at the beginning of the proof, we can write $\Theta f = \widehat{N}(\Theta)(-i\nu_-) f$ for every $f \in Dom(\Theta) \cap L^2(\partial M)$.
We also note that, using Equation \eqref{Eq: product of indicial families}, it holds 
\begin{equation*}
\widehat{N}(A_r^*A_r)(-i\nu_-) = \widehat{N}(\widetilde{A}_r)(-i\nu_-)^* \widehat{N}(A_r)(-i\nu_-)
\end{equation*}
where $\widetilde{A}_r = x^{2 \nu_-} A_r x^{-2\nu_-}$ and where the adjoint is computed with respect to the $L^2$-paring induced by the metric $h$ on $\partial M$, {\it cf.} Definition \ref{Def: asymptotically AdS}. Using these data, we can rewrite the boundary term as
\begin{equation*}
\begin{split}
& \langle \Theta\gamma_-u,\gamma_-(A_r^*A_r u) \rangle_{\partial M}  = 
\langle \widehat{N}^*(\widetilde{A}_r\Theta)(-i\nu_-)\gamma_-u, \widehat{N}(A_r)(-i \nu_-)\gamma_-u \rangle_{\partial M}= \\
&= \langle \gamma_- (\widetilde{A}_r \Theta) u, \gamma_-(A_r u) \rangle_{\partial M} = \langle \gamma_-(\Theta \widetilde{A}_r u + [\widetilde{A}_r,\Theta]u),\gamma_- A_r u \rangle_{\partial M}
\end{split}
\end{equation*}

Using Chauchy-Schwartz inequality and Lemma \ref{eq:gamma-bound} it holds 
\begin{equation*}
| \langle \gamma_-\Theta \widetilde{A}_r u,\gamma_- A_r u \rangle_{\partial M}|  \leq  C_1 \| \Theta \widetilde{A}_r u  \|_{\mathcal{L}^2(M)} + C_2 \| A_r u \|_{\mathcal{L}^2(M)}^2,
\end{equation*}
where $C_1$ and $C_2$ are suitable constants. If $\Theta\in\Psi^k_b(\partial M)$  with $k\leq 0$ it holds 
\begin{equation*}
\| \Theta \widetilde{A}_r u  \|_{\mathcal{L}^2(M)}^2 \leq \| \Theta \widetilde{A}_r u  \|_{\mathcal{H}^1(M)}^2 \leq \|\chi   \widetilde{A}_r u  \|_{\mathcal{H}^1(M)}^2
\end{equation*}
Thus, proceeding as in the previous case, using Lemma \ref{prop-Psi0_Bound} we arrive at the same estimate for $|\langle \Theta\gamma_-u,\gamma_-(A_r^*A_r u) \rangle|_{\partial M} |$.
Combining all the bounds together with Equation \eqref{eq:formbound}, we obtain the sought thesis.
\end{proof}

Now we focus on the case where $\Theta\in\Psi^k_b(\partial M)$ with $0 < k \leq 2$. As in Lemma \ref{Lem: lemma-bound-quad-form-v1} we extend $\Theta$ trivially to the whole $M$ by considering it independent from the coordinate $x$. In addition we observe that each $\Theta$ identifies per duality a map from $\mathcal{H}^1_{loc}(M)$ to $\dot{\mathcal{H}}^{-1}_{loc}(M)$.

\begin{lemma}\label{lemma-bound-quad-form-v2}
Let $U \subset M$ be a coordinate patch such that $U\cap\partial M\neq\emptyset$ and let $m \leq 0$. Let $\Theta \in \Psi_b^k(\partial M)$ with $0 < k \leq 2$ and let $\mathcal{A}=\{ A_r \ | \ r \in (0,1) \}$ be a bounded subset of $\Psi^s_b(M)$, $s\in\mathbb{R}$, with compact support in $U$, such that 
\begin{equation*}
A_r \in \Psi^m_b(M) \ \textit{for each} \ r \in (0,1)
\end{equation*}
Let $G_0 \in \Psi^s_b(M)$ and $G_1 \in \Psi_b^{s-1/2}(M)$ be elliptic on $WF^\prime_b(\mathcal{A})$, both with compact support in $U$. Then there exists $C_0 > 0$ and $\chi \in \mathcal{C}^\infty_0(U)$ such that
\begin{equation*}
\begin{split}
& \mathcal{E}_0(A_r u,A_r u) - \varepsilon \mathcal{Q}(A_r u,A_r u) \leq C_0 \Big( \|\chi u \|^2_{\mathcal{H}^{1,m}(M)} + \|\chi P_\Theta u \|^2_{\dot{\mathcal{H}}^{-1,m}(M)} + \\
& + \| G_0 P_\Theta u\|^2_{\dot{\mathcal{H}}^{-1}(M)}+ \| G_0 \Theta u\|^2_{\dot{\mathcal{H}}^{-1}(M)} +  \|\chi \Theta u \|^2_{\dot{\mathcal{H}}^{-1,m}(M)}  +  \|G_1 u \|^2_{\mathcal{H}^1(M)} \Big)
\end{split}
\end{equation*}
for every $r \in (0,1)$ and every $u \in \mathcal{H}^{1,m}_{loc}(M)$, provided that the following conditions are met:
$$\ WF^{-1,s}_b(P_\Theta u) \cap WF^\prime_b(G_0) = \emptyset$$
$$\ WF^{-1,s}_b(\Theta u) \cap WF^\prime_b(G_0) = \emptyset$$
$$ WF_b^{1,s-1/2}(u) \cap WF_b^\prime(G_1) = \emptyset $$
\end{lemma}
\begin{proof}
The proof is analogous to that of Lemma \eqref{Lem: lemma-bound-quad-form-v1}, hence we do not enter into the details. We point out that the only key difference is the estimate of the boundary term $\langle \Theta\gamma_-u,\gamma_-(A^*_r A_r u) \rangle|_{\partial M}$.
This time, thanks to the inclusions $\mathcal{H}^1_{loc}(M) \subset \mathcal{L}^2_{loc}(M) \subset \dot{\mathcal{H}}^{-1}_{loc}(M)$ we can control the boundary term as
\begin{equation}
\begin{split}
\langle \gamma_- (\widetilde{A}_r \Theta) u, \gamma_-(A_r u) \rangle |_{\partial M} \leq C_0 \left(  \|\widetilde{A}_r  \Theta  u  \|_{\mathcal{L}^2(M)} + \| A_r u \|_{\mathcal{L}^2(M)}^2 \right) \\
\leq C_1 \left(  \| \widetilde{A}_r  \Theta  u  \|_{\dot{\mathcal{H}}^{-1}(M)}  + \| A_r u \|_{\mathcal{L}^2(M)}^2 \right)\\
 \leq  C \left( \|\chi u \|^2_{\mathcal{H}^{1,m}(M)}+  \|G_1 u \|^2_{\mathcal{H}^1(M)}  + \| G_0 \Theta u\|^2_{\dot{\mathcal{H}}^{-1}(M)} +  \|\chi \Theta u \|^2_{\dot{\mathcal{H}}^{-1,m}(M)}  \right)
\end{split}
\end{equation}
\end{proof}

\begin{rem}
Using that, for $\epsilon>0$, $\Psi_b^m(M) \subset \Psi^{m+\varepsilon}_b(M)$, the previous results holds true also for $G_1 \in \Psi_b^{s-1/2}(M)$, similarly to what happens in \cite{GaWr18}.
\end{rem}

\begin{rem}\label{Rem: obstruction for large k}
	Observe that, if we would have allowed $k$ to be larger than $2$, we would haven not been able to prove in general a result similar to Lemma \ref{lemma-bound-quad-form-v2}. For this reason we have discarded such scenario. We stress that, at the level of applications, this is a mild constraint since, to the best of our knowledge, interesting examples of boundary conditions, such as the Robin ones discussed in \cite{GaWr18} or those of Wentzell type, see \cite{Zahn:2015due} are all included in the regime $k\leq 2$.
\end{rem}

\section{Propagation of singularities}\label{Sec: propagation of singularities}

In this section we present the main result of this work, namely we derive a theorem of propagation of singularities for the Klein-Gordon operator with boundary conditions ruled by $\Theta\in\Psi^k_b(\partial M)$, $k\leq 2$, as discussed in Subsection \ref{Sec: Boundary conditions and the associated Dirichlet form}.

\subsection{The compressed characteristic set}\label{Sec: compressed characteristic set}

We start from a detailed analysis of the characteristic set of the principal symbol of the Klein-Gordon operator for the case in hand. We recall that the principal symbol of $x^{-2}P$, see Equation \eqref{Eq: KG on M} is $\widehat{p}\doteq\widehat{g}(X,X)$, where $X\in\Gamma(T^*M)$. The associated {\em characteristic set} is 
\begin{equation}\label{Eq: characteristic set}
\mathcal{N}=\left\{(q,k_q)\in T^*M\setminus\{0\}\;|\; \widehat{g}^{ij}(k_q)_i (k_q)_j=0\right\},
\end{equation}
while the {\em compressed characteristic set} is 
\begin{equation}\label{Eq: compressed characteristic set}
\dot{\mathcal{N}}=\pi[\mathcal{N}]\subset{}^b\dot{T}(M),
\end{equation}
where $\pi$ is the projection map from $T^*M$ to the compressed cotangent bundle, {\it cf.} Equation \eqref{Eq: compressed b-cotangent bundle}. We equip $\dot{\mathcal{N}}$ with the subspace topology inherited from ${}^b T^* M$. In addition, it is convenient to individuate in the compressed b-cotangent bundle the following three conic subsets:

\vskip.2cm

\begin{itemize}
\item The {\em elliptic} region 
\begin{equation}\label{Eq: elliptic region}
\mathcal{E}(M) = \{ q \in{}^b\dot{T}^*M \setminus\{0\} \ : \ \pi^{-1}(q) \cap \mathcal{N} = \emptyset \},
\end{equation}
where $\pi:T^*M\to{}^b\dot{T}^*M$.

\vskip.2cm

\item The {\em glancing} region  
\begin{equation}\label{Eq: glancing region}
\mathcal{G}(M) = \{ q \in{}^b\dot{T}^*M \setminus\{0\} \ : \ Card( \pi^{-1}(q) \cap \mathcal{N}) = 1 \},
\end{equation}
where $Card$ refers to the cardinality of a set.

\vskip .2cm

\item The {\em hyperbolic} region 
\begin{equation}\label{Eq: hyperbolic region}
\mathcal{H}(M) = \{ q \in{}^b\dot{T}^*M \setminus\{0\} \ : \ Card( \pi^{-1}(q) \cap \mathcal{N}) = 2 \}.
\end{equation}
\end{itemize}

\vskip .2cm

\begin{rem}\label{Rem: nature of different points}
Consider now $\widetilde{q}\in{}^bT^*\partial M$ such that $\widetilde{q}=(0,y_i,0,\eta_i)$, $i=1,\dots,n-1$, where we used the same coordinates introduced in Section \ref{Sec: Geometric preliminaries}. It descends that $\pi^{-1}(\widetilde{q})=(0,y_i,\xi,\eta_i)$ where $(\xi,\eta_i)\in T^*_{(0,y_i)}M$ which entails that $\pi^{1-}(\widetilde{q})\simeq\mathbb{R}$. 
Considering Equation \eqref{Eq: characteristic set} together with Equation \eqref{Eq: metric near the boundary}, we can infer that $\pi^{-1}(\widetilde{q})\cap\mathcal{N}$ corresponds to solving the algebraic equation $\xi^2+h^{ij}\eta_i \eta_j = 0$. This entails that a point $\widetilde{q} \in {}^bT^*M$ lies in $\mathcal{H}$ when $h^{ij}\eta_i \eta_j<0$, in $\mathcal{G}$ when $h^{ij}\eta_i \eta_j=0$ and in $\mathcal{E}$ when $h^{ij}\eta_i \eta_j>0$.
\end{rem}

\begin{defn}\label{Def: generalized broken bicharacteristics}
Let $I \subset \mathbb{R}$ be an interval. A continuous map $\gamma : I \rightarrow \dot{\mathcal{N}}$ is a {\em generalized broken bicharacteristic} (GBB) if for every $s_0 \in I$ the following conditions hold:
\begin{itemize}
\item[a)] If $q_0 = \gamma(s_0) \in \mathcal{G}$, then for every $\omega\in \Gamma^\infty(^bT^*M)$,
\begin{equation}
\frac{d}{ds}(\omega \circ \gamma) = \{ \widehat{p},\pi^* \omega \}(\eta_0)
\end{equation}
where $\eta_0 \in \mathcal{N}$ is the unique point for which $\pi(\eta_0)=q_0$, while $\pi:T^*M\to{}^bT^*M$ and $\{,\}$ are the Poisson brackets on $T^*M$.
\item[b)] If $q_0 = \gamma(s_0) \in \mathcal{H}$, then there exists $\varepsilon > 0$ such that $0 < |s-s_0| < \varepsilon $ implies $x(\gamma(s))\neq 0$, where $x$ is the global boundary function, {\it cf.} Definition \ref{Def: asymptotically AdS}.
\end{itemize}
\end{defn}

\begin{rem}\label{Rem: significance of the previous definition}
Observe that, since $\gamma\in C^0(I;\dot{\mathcal{N}})$ the component of the co-vector tangent to the boundary is conserved. The first condition tells us heuristically that in the glancing region, GGBs are integral curves of the Hamilton vector field associated with the principal symbol $\widehat{p}$. The second condition implies instead that, at hyperbolic points, GBBs reflect instantaneously. Hence, a GBB coming from $\mathring{M}$ propagates along the boundary only at glancing points.
\end{rem}

In the following, we outline a few distinguished properties of GBBs. The next lemma summarizes results from both \cite{Leb97} and \cite{Vasy08}.

\begin{lemma}\label{lemm:gbb-lebau}
Let $\mathcal{R}_K[a,b]$ be the space of the generalized broken bicharacteristics $\gamma : [a,b] \rightarrow K$ where $K \subset \dot{\mathcal{N}}$ is compact.
Let ${\gamma_n}$ be a sequence in $\mathcal{R}_K[a,b]$ converging uniformly to a curve $\gamma$. Then $\gamma : [a,b] \rightarrow K$ is a generalized broken bicharacteristic.
In addition, if $\mathcal{R}_K[a,b]$ is not empty, then it is compact in the uniform topology.
\end{lemma}

To conclude this subsection, we focus our attention on the boundary $\partial M$. Let us consider once more a chart $U\subset M$ such that $U\cap\partial M\neq \emptyset$. Following the conventions explained in Section \ref{Sec: Geometric preliminaries} we consider on $T^*_UM$ coordinates $(x,y_i,\xi,\eta_i)$, $i=1,\dots,n-1$, where we identify the time coordinate with $y_{n-1}$, while $\eta_{n-1}$ is the associated dual coordinate. With these premises the following lemma holds true. Observe that the proof is identical to that of \cite[Lemma 6.2]{GaWr18} with the due exception that we have to take into account the specific form of the metric on $\partial M$, {\it cf.} Remark \ref{Rem: boundary metric}. Yet since the function $\beta$ is bounded and strictly positive on $U$, it plays no specific role.

\begin{lemma}\label{lemma:elliptic-cases}
If $q_0 \in  {}^bT_U^*M \setminus\{0\}$, there exists a conic neighborhood $V$ of $q_0$ in which one of the following facts is true:
\begin{itemize}
\item[1) ] If $q_0 \in{}^b\dot{T}^*M$, there exists $\varepsilon > 0$ such that $\sigma^2 < \varepsilon^2(\beta\eta^2_{n-1}+\kappa^{ij}\eta_i\eta_j)$ and $\kappa^{ij}\eta_i \eta_j > \beta \eta^2_{n-1}$.
\item[2) ] If $q_0 \not\in{}^b\dot{T}^*M$, there exists $C>0$ such that $|\eta_{n-1}|<C|\sigma|$
\end{itemize}
\end{lemma}

\begin{rem}
For simplicity, in the following we shall work with pseudodifferential operators whose compact support is contained in a fixed local chart. However, our results are also valid in the general case in which the support is not contained in one coordinate patch, using a partition of unity argument.
\end{rem}

\subsection{Estimates in the elliptic region}\label{Sec: elliptic region}

In this part of the section we start the analysis aimed at deriving suitable microlocal estimates which will be necessary to prove a propagation of singularity theorem. The analysis will be divided in three parts, one for each of the regions individuated above. In each case we discuss separately the scenarios in which $\Theta\in\Psi^k_b(\partial M)$ with $0<k\leq 2$ or with $k\leq 0$. In addition, we  assume implicitly that we are always considering the trivial extension of $\Theta$ to $M$, {\it i.e.} constant in the coordinate $x$. 

As the title of the subsection suggests, we start from $\mathcal{E}(M)$ as in Equation \eqref{Eq: elliptic region}. As above, we consider a coordinate neighbourhood $U\subset M$ and we indicate with $T^*_UM\doteq T^*M|_U$ and ${}^bT^*_UM\doteq{}^bT^*M|_U$. In addition, using the same coordinates discussed in Section \ref{Sec: Geometric preliminaries}, in full analogy to Equation \eqref{Eq: Q0} and \eqref{Eq: Q1}, we introduce the operators 
\begin{equation}\label{Eq: Q-operators}
Q_0=F\nabla_x F^{-1},\quad Q_i=F\nabla_i F^{-1},\; i=1,\dots,n-1
\end{equation}
where $F$ is an admissible twisting function, {\it cf.} Definition \ref{Def: twisting function} and Definition \ref{Def: admissible twisting function}. Recalling Definition \ref{Def: wavefrontset for Hkm-loc} it holds, 

\begin{prop}[microlocal elliptic regularity]\label{Prop: elliptic regularity}
Let $u \in \mathcal{H}^{1,m}_{loc}(M)$ for $m \leq 0$ and let $q_0 \in {}^bT^*_UM$. If $s \in \mathbb{R} \cup \{+\infty\}$ and if $\Theta\in\Psi^k_b(\partial M)$ with $0<k\leq 2$, then $q_0\in WF_b^{1,s}(u) \setminus \left(WF_b^{-1,s}(P_\Theta u) \cup WF_b^{-1,s}(\Theta u)\right)$ entails $q_0\in\dot{\mathcal{N}}$.
\end{prop}

\begin{proof}
We follow the strategy of \cite[Th. 3]{GaWr18} with the due difference that we need to control the contribution due to $\Theta$. Hence we proceed by induction with respect to $s$, proving that $q \not \in WF_b^{1,s+1/2}(u)$ and $q \not \in WF_b^{-1,s}(P_\Theta u)\cup WF_b^{-1,s}(\Theta u)$ entails $q \not \in WF_b^{1,s}(u)$. 

The statement holds true for $s \leq m+1/2$ since $u \in \mathcal{H}^{k,m}(M)$. To proceed in the inductive procedure, observe that, since we want to study properties of the wavefront set at a point $q_0\in {}^bT^*_U M$ it is convenient to evaluate the energy form, {\it cf.} Equation \eqref{Eq: twisted Dirichlet form} with the arguments replaced by $Au$, with $A \in \Psi^s_b(M)$ elliptic at $q_0$ and with compact support in $U \cap \{ x < \delta \}$ where $\delta>0$. To control such energy form we consider a family $\{ J_r\in\Psi^{m-s-1}_b(M)\, |\, r \in (0,1) \}$, bounded in $\Psi^0_b(M)$ converging to the identity in $\Psi^1_b(M)$ as $r \rightarrow 0$. We approximate $A$ using the family $\mathcal{A}=\{A_r = J_r A\}$. As shown in \cite{Vasy10}, it holds 
\begin{equation}\label{eq:elliptic-energy1}
\begin{split}
\mathcal{E}_0(A_r u, A_r u) \geq \| Q_0 A_r u \|^2 +\\
+ (1-C\delta) \langle \kappa^{ij} Q_i A_r u , Q_j A_r u \rangle, - (1+C\delta) \|\beta^{\frac{1}{2}} Q_{n-1} A_r u\|^2,
\end{split}
\end{equation}
where $\kappa_{ij}$ and $\beta$ are the components of the metric as in Theorem \ref{Th: globally hyperbolic}, while $C$ is a positive constant. In addition we have adopted the convention that $y_{n-1}$ corresponds to the time coordinate $\tau$ on the boundary, see Remark \ref{Rem: boundary metric} while $\eta_{n-1}$ is the associated momenta on the $b$-cotangent bundle.  
It is convenient to distinguish two cases, corresponding to those of Lemma \ref{lemma:elliptic-cases} First, let us assume that $q \in {}^b\dot{T}^*M$.  We can rewrite the last two terms of Equation \eqref{eq:elliptic-energy1} as 
\begin{gather}
\langle \left[ (1-C\delta)\kappa^{ij}Q^*_i Q_j - (1+C\delta)\beta Q^*_{n-1} Q_{n-1} \right]A_r u , A_r u \rangle + \notag \\
+ \langle\left( (1-C\delta) \left( Q^*_j \kappa^{ij} \right) Q_i - (1+C\delta) \left(Q^*_{n-1}\beta\right) Q_{n-1}\right) A_r u , A_r u \rangle
\end{gather}
Now we focus on the operator $(1-C\delta)\kappa^{ij}Q^*_i Q_j - (1+C\delta)\beta Q^*_{n-1} Q_{n-1}$, whose symbol $(1-C\delta)\kappa^{ij}\eta_i \eta_j - (1+C\delta)\beta \eta^2_{n-1}$ is of order 2. Since, whenever $q \in {}^b\dot{T}^*M$, it holds $\kappa^{ij}\eta_i\eta_j > (1+\varepsilon)\beta\eta_{n-1}^2$, {\it cf.} Lemma \ref{lemma:elliptic-cases}, 
\begin{equation*}
\begin{split}
(1-C\delta)\kappa^{ij}\eta_i\eta_j - (1+C\delta)\beta \eta_{n-1}^2 = \\ (1-C\delta)(\kappa^{ij}\eta_i\eta_j-\beta\eta_{n-1}^2)-2C\delta\beta \eta_{n-1}^2 
>  \Big( \varepsilon (1-C\delta)-2C\delta \Big)\beta\eta_{n-1}^2 
\end{split}
\end{equation*}
Then, for $C$ and $\delta$ small enough, it holds:
\begin{equation*}
\begin{split}
(1-C\delta)\kappa^{ij}\eta_i\eta_j - (1+C\delta)\beta \eta_{n-1}^2 > \frac{\varepsilon}{2}\beta \eta_{n-1}^2
\end{split}
\end{equation*}
This inequality yields that $(1-C\delta)\kappa^{ij}\eta_i\eta_j - (1+C\delta)\beta \eta^2_{n-1}$ is a positive and elliptic symbol at $q$. Therefore, we can take an approximate square root $R \in \Psi^1_b(M)$ of the operator $(1-C\delta)\kappa^{ij}Q^*_i Q_j - (1+C\delta)\beta Q^*_{n-1} Q_{n-1}$, namely a pseudodifferential operator with principal symbol given by $\sigma_{b,1}(R) = (1-C\delta)\kappa^{ij}\eta_i\eta_j - (1+C\delta)\beta \eta^2_{n-1}$ and such that 
$$R^2 = (1-C\delta)\kappa^{ij}Q^*_i Q_j - (1+C\delta)\beta Q^*_{n-1} Q_{n-1} + S, $$
with $S \in \Psi^{-\infty}_b(M)$. To summarize, we can recast
\begin{equation*}
\begin{split}
(1-C\delta) \langle \kappa^{ij} Q_i A_r u , Q_j A_r u \rangle - (1+C\delta) \|\beta^{\frac{1}{2}} Q_{n-1} A_r u\|^2 
\end{split}
\end{equation*}
as
\begin{equation}
\begin{split}
\langle R A_r u , R A_r u \rangle + \langle T A_r u , A_r u \rangle
\end{split}
\end{equation}
with $T = S +  (1-C\delta) \left(Q^*_i \kappa^{ij}\right) Q_j - (1+C\delta) \left(Q^*_{n-1}\beta\right) Q_{n-1} \in \Psi_b^1(M)$. 
Since $T \in \Psi_b^1(M)$ it descends that $|\langle T A_r u, A_r u \rangle|$ is uniformly bounded for $r \in (0,1)$. Let $\Lambda_+ 	\in \Psi^{1/2}_b(M)$ be an elliptic pseudodifferential operator and let $\Lambda_- \in \Psi^{-1/2}_b(M)$ be a parametrix. Then  $\mathbb{I} = \Lambda_- \Lambda_+ + E$, with $E \in \Psi^{-\infty}_b(M)$ and we can write:
\begin{equation}
\langle T A_r u, \mathbb{I} A_r u \rangle = \langle  \Lambda_-^* T A_r u, \Lambda_+ A_r u \rangle + \langle T A_r u , E A_r u \rangle 
\end{equation}
By Cauchy-Schwartz and triangular inequalities, it descends
\begin{equation}
|\langle T A_r u, \mathbb{I} A_r u \rangle|^2 \leq  \| \Lambda_-^* T A_r u \| \cdot \| \Lambda_+ A_r u  \| + \| T A_r u \| \cdot \| E A_r u \| 
\end{equation}
Thanks to Lemma \ref{prop-boundL2} and to the hypotheses on $u$ and on the family $\{A_r\}$, all norms on the right hand side are uniformly bounded for $r \in (0,1)$. In particular it holds
\begin{equation}
|\langle T A_r u, \mathbb{I} A_r u \rangle|^2 \leq C \left( \| G_1 u \|^2_{\mathcal{H}^1(M)} + \| G_2 u \|^2_{\mathcal{H}^1(M)} + \|\chi u \|^2_{\mathcal{H}^{1,m}(M)} \right)
\end{equation}
where $G_1 \in \Psi^{s-1/2}_b(M)$ is such that $WF^\prime(\Lambda_-^* T A_r) \cup WF^\prime(\Lambda_+ A_r) \subset ell_b(G_1)$ and $G_2 \in \Psi^{s-1}_b(X)$ is such that $WF^\prime(T A_r) \subset ell_b(G_2)$.
Therefore from Equation \eqref{eq:elliptic-energy1} one obtains
\begin{equation}
0 \leq (1-C\delta)\| Q_0 A_r u \|^2_{\mathcal{L}^2(M)} + \| R A \|^2_{\mathcal{L}^2(M)} \leq \mathcal{E}_0(A_r u , A_r u) - \langle T A_r u, A_r u \rangle
\end{equation}
Note that the Dirichlet form $\mathcal{E}_0(A_r u , A_r u)$ is uniformly bounded for $r \rightarrow 0$ thanks to Lemma \ref{lemma-bound-quad-form-v2}. Thus, we can draw the same conclusion for $(1-C\delta)\| Q_0 A_r u \|^2_{\mathcal{L}^2(M)} + \| R A \|^2_{\mathcal{L}^2(M)} $. Hence one can find subsequences $A_{r_k}u$, $Q_0 A_{r_k}u$ and $ R A_{r_k}u$, weakly convergent in $\mathcal{L}^2(M)$ and such that $r_k \rightarrow 0$ as $k \rightarrow \infty$. Since they converge to $Au$, $Q_0 Au$ and $R Au$ in $\mathcal{D}^\prime(M)$, in particular the weak limits lie in $\mathcal{L}^2(K)$ with $K$ a compact subset of $M$ such that $K \cap (U \cap \{x < \delta \}) \neq \emptyset$. This entails that $Au \in \mathcal{H}^1(K)$, and hence that $q \not \in WF^{1,s}_b(u)$.\\

As for the second case of Lemma \ref{lemma:elliptic-cases}, first we note that for $u$ supported in $\{x<\delta\}$, the following relation holds true
$$\| Q_0 u \|^2_{\mathcal{L}^2(M)} \geq \delta^{-2} \| x Q_0 u\|^2_{\mathcal{L}^2(M)}$$
Hence it holds 
\begin{gather*}
\mathcal{E}_0(A_r u,A_r u) \geq  \delta^{-2}  \langle x Q_0 A_r u, x Q_0 A_r u \rangle + \\
\langle [(1-C\delta)\kappa_{ij}Q^iQ^j- (1+C\delta)\beta Q^2_{n-1}]A_r u, A_r u \rangle + \langle T A_r u, A_r u \rangle,
\end{gather*}
where $T$ accounts for lower order terms. We can rewrite the right hand side as
\begin{equation}
\begin{split}
\langle [ \delta^{-2} (x Q_0)^* (x Q_0) - (1+C\delta\beta)Q^2_{y_{n-1}} ] A_r u,A_r u \rangle + \\
\langle (1-C\delta)(\kappa)_{ij}Q^iQ^j A_r u, A_r u \rangle + \langle T A_r u, A_r u \rangle
\end{split}
\end{equation}
The operator $ \delta^{-2} (x Q_0)^* (x Q_0) - (1+C\delta)\beta Q^2_{n-1} $ has symbol $\zeta^2/(2\delta^2)-(1+C\delta)\beta\eta^2_{n-1}$, that is elliptic near $V$ since, on account of Lemma \ref{lemma:elliptic-cases}, there must exist a constant $c$ such that 
$$ \frac{\zeta^2}{2\delta^2}-(1+C\delta)\beta\eta^2_{n-1} >c\beta \eta^2_{n-1}$$
Hence, we can define, modulo lower order terms, its square root as a pseudodifferential operator and then we proceed exactly like in the previous case.
\end{proof}

If we consider $\Theta\in\Psi^k_b(\partial M)$ with $k \leq 0$, we can prove a statement similar to the preceding one using Lemma \ref{Lem: lemma-bound-quad-form-v1} instead of Lemma \ref{lemma-bound-quad-form-v2}. For this reason we omit to give a detailed proof and we limit ourselves to reporting the final statement:

\begin{prop}[microlocal elliptic regularity]\label{prop-mer}
Let $u \in \mathcal{H}^{1,m}_{loc}(M)$ for some $m \leq 0$ and let $q_0 \in  {}^bT^*_UM$. If $s \in \mathbb{R} \cup \{+\infty\}$ and if $\Theta\in\Psi^k_b(\partial M)$ with $k\leq 0$, then $WF_b^{1,s}(u) \setminus \dot{\mathcal{N}} \subseteq WF_b^{-1,s}(P_\Theta u)$.
\end{prop}

\subsection{Estimates in the hyperbolic region}\label{Sec: hyperbolic region}

We focus our attention on the hyperbolic region $\mathcal{G}(M)$ introduced in Equation \eqref{Eq: hyperbolic region}, deriving suitable microlocal estimates which will be used in the proof of a propagation of singularity theorem. In comparison to the previous case, we adopt a different strategy mainly based on a positive commutator argument. 

\medskip

As a preliminary step we observe that if $u \in \mathcal{H}^{1,m}_{loc}(M)$ and if $A\in\Psi^0_b(M)$ with principal symbol $\sigma_{b,0}(A)=a$ has compact support, then a direct computation yields 

\begin{gather}\label{eq:im-energy}
2i Im \mathcal{E}_0 (u,A^* A u) = \langle \widehat{g}^{ij} Q_j u , [Q_i, A^* A] u  \rangle - \langle [\widehat{g}^{ij}  Q_j , A^* A] u, Q_i u \rangle + \notag \\ + \langle Q_0 u , [Q_0, A^* A] u \rangle - \langle [Q_0,A^* A] u, Q_0 u \rangle + \langle [Q_i \widehat{g}^{ij}Q_j, A^* A] u, u \rangle,
\end{gather}
where the operators $Q_i$, $i=0,\dots,n-1$ are defined as in Equation \eqref{Eq: Q-operators}. For future convenience, it is useful to compute explicitly the commutators in the first two terms in the second line, getting:
\begin{gather}
\langle Q_0 u , [Q_0, A^* A] u \rangle - \langle [Q_0,A^* A] u, Q_0 u \rangle =\notag \\ =\langle Q_0 u, Q_0 A_1 u \rangle - \langle Q_0 A_1 u, u \rangle +  
+ \langle Q_0 u, A_0 u \rangle - \langle A_0 u, Q_0 u \rangle,
\end{gather}
where $A_0\in\Psi^0_b(M)$, $A_1\in\Psi^{-1}_b(M)$ have as principal symbol respectively $a_0=-i \partial_x a^2$ and $a_1 = -i \partial_\zeta a^2$. 

\medskip

\noindent We focus on proving the desired information on the wavefront sets in the hyperbolic regions. We divide the analysis in two parts depending whether $\Theta\in\Psi^k_b(\partial M)$ with $k\leq 0$ or with $0<k\leq 2$. In the following we consider implicitly the trivial extension of $\Theta$ to the whole $M$ employing with a slight abuse of notation the same symbol. Furthermore we shall use the same coordinates introduced in Section \ref{Sec: Geometric preliminaries} with the implicit convention that $y_{n-1}$ coincides with $\tau$, {\it cf.} Theorem \ref{Th: globally hyperbolic} and Remark \ref{Rem: boundary metric} while $\eta_{n-1}$ is the associated momentum on the $b$-cotangent bundle.

\begin{prop}\label{prop:hyp_reg_pos}
Let $\Theta \in \Psi^k_b(M)$ with $0 < k \leq 2$. Let $u \in \mathcal{H}^{1,m}_{loc}(M)$ with $m \leq 0$ and suppose that $q_0 \not \in WF_b^{-1,s+1}(P_\Theta u)\cup WF_b^{-1,s+1}(\Theta u)$. If there exists a conic neighborhood $W \subset T^*M \setminus\{0\}$ of $q_0$ such that $W \cap \{ \zeta < 0 \} \cap WF_b^{1,s}(u) = \emptyset $ then $q_0 \not \in WF_b^{1,s}(u)$ 
\end{prop}

\begin{prop}\label{prop:hyp_reg_negnull}
Let $\Theta \in \Psi^k_b(M)$ for some $k \leq 0$. Let $u \in \mathcal{H}^{1,m}_{loc}(M)$ for some $m \leq 0$ and suppose that $q_0 \not \in WF_b^{-1,s+1}(P_\Theta u)\cup WF_b^{-1,s+1}(\Theta u)$. If there exists a conic neighborhood $W \subset{}^bT^*M \setminus\{0\}$ of $q_0$ such that $W \cap \{ \zeta < 0 \} \cap WF_b^{1,s}(u) = \emptyset $ then $q_0 \not \in WF_b^{1,s}(u)$.
\end{prop}

The proof of both Proposition \ref{prop:hyp_reg_pos} and \ref{prop:hyp_reg_negnull} is similar to that of 
Proposition \ref{Prop: elliptic regularity} and \ref{prop-mer} respectively, the main difference consisting in replacing Lemma \ref{Lem: lemma-bound-quad-form-v1} and \ref{lemma-bound-quad-form-v2} with suitable counterparts tied to the hyperbolic region. For this reason, first we discuss these counterparts postponing the proof to the end of the section.

\begin{rem}
Let $\Theta \in \Psi^k_b(M)$ for some $0 < k \leq 2$ and let $Z \subset W$ with $q_0 \in Z$. Since $q_0 \not \in WF_b^{-1,s+1}(P_\Theta u) \cup WF_b^{-1,s+1}(\Theta u)$, if $Z$ is small enough then, by the elliptic regularity theorem, {\it cf.} Prop. \ref{Prop: elliptic regularity},
$$\left( WF_b^{-1,s+1}(P_\Theta u) \cup WF_b^{-1,s+1}(\Theta u) \right) \cap Z = \emptyset. $$ 
Hence we can conclude that $Z \cap WF_b^{1,s}(u) \subset \dot{\mathcal{N}}$. In particular, this fact means that on the set $Z \cap \{ \zeta < 0 \} \cap WF_b^{1,s}(u)$ it holds $x \neq 0$ and a point $q_0 \in WF_b^{1,s}(u)$ can be seen as a limit of points in the wavefront set, each of which does not lie on the boundary. An analogous statement holds true for the case in which $k \leq 0$.
\end{rem}

\medskip

Let $U$ be a coordinate patch such that $U\cap\partial M\neq\emptyset$ and let $q_0 \in \mathcal{H}(M) \cap { }^b T_U^* M$.  Following \cite{Vasy08} the first step consists of introducing the smooth scalar function on ${}^bT^*M$, $\mu=-\zeta=-x\xi$ which enjoys the notable properties that it is homogeneous of degree $0$ and that, in a neighborhood of $q_0$, the sign of $H_{\widehat{p}}\left(\pi^*\mu\right)$ does not change. Here $H_{\widehat{p}}$ is the Hamiltonian vector field associated to the principal symbol $\widehat{p}$ of $x^{-2}P$. 

If we consider the $b$-cosphere bundle ${}^bS^*M$ as per Equation \eqref{Eq: cosphere bundle} together with the associated coordinates on ${}^bS^*_UM\doteq{}^bS^*M|_U$, we  can introduce the function $\widehat{\omega}:{}^bS^*_UM\to\mathbb{R}$
\begin{equation}\label{Eq: function on the cosphere bundle}
\widehat{\omega}(q) = |x(q)|^2+\sum\limits_{i=1}^{n-2}|y_i(q)-y_i(q_0)|^2+|\widehat{\zeta}(q)-\widehat{\zeta}(q_0)|^2+\sum\limits_{i=1}^{n-2}|\widehat{\eta}_i(q)-\widehat{\eta}_i(q_0)\Big|^2,
\end{equation}
which induces in turn a function $\omega:{}^bT^*_UM\setminus\{0\}\to\mathbb{R}$ defined as $\omega=\widehat{\omega}\circ\pi_S$ where $\pi_S:{}^bT^*M\setminus\{0\}\to{}^bS^*M$ is the natural projection map implementing the quotient in Equation \eqref{Eq: cosphere bundle}. Observe that, for the sake of simplicity of the notation, we have refrained from indicating that $\widehat{\omega}$ depends explicitly from the choice of $q_0$. In addition, on a conic neighborhood of $q_0$, consider the homogeneous smooth function $\phi$
\begin{equation}
\phi = \mu + \frac{1}{\lambda^2\delta}\omega
\end{equation}
 where $\lambda$ and $\delta$ are positive parameters. By construction $\phi$ can be read as a $\pi$-invariant function on $T^*M\setminus\{0\}$ and, to localize it near $q_0$, consider $\chi_0,\chi_1,\in\mathcal{C}^\infty(\mathbb{R})$ such that 
$$\chi_1(s)=\left\{\begin{array}{ll}
0 & \textrm{if}\; s\in (-\infty,0)\\
1 & \textrm{if}\; s\in [1,\infty)
\end{array}
\right.,$$
while the derivative $\chi^\prime_1$ is positive on $(0,1)$. At the same time
\begin{equation}
\chi_0(s) =
\begin{cases}
0 & \textit{ if } s \leq 0 \\
\exp(-s^{-1}) & \textit{ if } s > 0
\end{cases}
\end{equation}
Consider now 
\begin{equation}\label{eq:hyp_localization}
a\doteq\chi_0(2-\phi/\delta)\chi_1(\hat{\zeta}/\delta+2),
\end{equation}
which is a smooth homogeneous function of degree zero in a conic neighborhood of $q_0$. On account of the properties of $\chi_0$ and of $\chi_1$ it holds that 

$$\omega \leq \lambda^2 \delta (2\delta - \eta) \leq 4 \delta^2 \lambda^2\quad\textrm{and}\quad |\widehat{\zeta}| \leq 2 \delta.$$

This entails that, for any $\lambda>0$ and for $\delta >0$ small enough, $f$ has support inside a conic neighborhood of $q_0$. At last, we also localize in a conic neighborhood of $q_0$ with compact closure and such that $\widehat{g}^{ab} k_a k_b > 0$ where $k_a=(\zeta,\eta_i)$, $i=1,\dots,n-1$ are coordinates on the fiber of the $b$-cotangent bundle . Let $V_0$ be a set satisfying these properties and consider a function $\psi_0 \in \mathcal{C}_0^\infty({}^b S^* M)$ such that $\psi_0=1$ on $V_0$ and whose support lies in a small neighborhood of $V_0$.

Now we can choose a family of pseudodifferential operators for regularization purposes. Let $\{ J_r | r \in (0,1) \}$ be a family of $\Psi$DOs in $\Psi^{s+k+1/2}_b(M)$ such that $J_r \in \Psi^m_b(M)$ for $r \in (0,1)$ and whose principal symbol is $j_r = \psi_0 \rho^{s+1/2}(1+r\rho)^{m-s-1/2}$. By construction $J_r$ is elliptic in $V_0$. We build a family of regulators
\begin{equation}
A_r = A J_r
\end{equation}
with $A\in \Psi^0_b(M)$ with principal symbol $a$ as in \eqref{eq:hyp_localization}. Note that since $A \in \Psi^0_b(M)$, $A_r$ is bounded in $\Psi^{s+1/2}_b(M)$. We report now a notable result \cite[Lemma 6.7]{GaWr18}

\begin{lemma}\label{lemma:commutator-positive}
Let $G \in \Psi^k_b(M)$. Given $\lambda > 0$ there exists $\delta_0 > 0$ such that for each $\delta \in (0,\delta_0)$
\begin{equation}
i [A^*_r A_r,G] = B_r^* D_r B_r + F_r +T_r
\end{equation}
where 
\begin{itemize}
\item $B_r \in \Psi^{s+1}_b(M)$, $r \in (0,1)$, has principal symbol $b_r = j_r b$ with 

$$b = \rho^{-1/2}\delta^{-1/2} \big[ \chi_0^\prime (2-\phi/\delta) \chi_0 (2-\phi/\delta) \big]^{1/2} \chi_1 (2 + \zeta/\delta),$$

\item $D_r \in \Psi^{k-2}_b(M)$, $r \in (0,1)$ and its principal symbol $d_r$ satisfies 
$$ \rho^{2-k} |(d_r)|\, \leq\, C_0 (\lambda \delta + \delta + \lambda^{-1}),$$
for some positive real constant $C_0$.

\item $T_r \in \Psi^{2s+k-1}(M)$, $r \in (0,1)$, is such that:

 $$WF_b^\prime(\mathcal{T}) \subset \{|\widehat{\zeta}| \leq 2 \delta, \omega^{1/2} \leq 2 \lambda \delta  \} $$
where $\mathcal{T} = \{ T_r | r \in (0,1) \}$.

\item $F_r \in \Psi^{2s+k}(M)$, $r \in (0,1)$, is such that 

$$WF_b^\prime(\mathcal{F}) \subset \{ -2 \delta \leq \widehat{\zeta} \leq -\delta, \omega^{1/2} \leq 2 \lambda \delta  \}, $$
where $\mathcal{F} = \{ F_r\; |\; r \in (0,1) \}$ is bounded in $\Psi^{2s+k}_b(M)$. 
\end{itemize}
\end{lemma}

As mentioned at the beginning of this section, we seek $Q\in\Psi^s_b(M)$ such that the norm of $Q A_r u$ is bounded in $\mathcal{L}^2(M)$. This can be individuated as follows. Starting from Proposition $\ref{eq:im-energy}$ we observe that $Im \mathcal{E}_0(u,A^*_r A_r u)$ contains a term of the form 

$$\langle Q_0 u, i Q_0 A_{1,r}u  \rangle.$$

Focusing on $a_{1,r}=-i \partial_\sigma a_r^2$, the principal symbol of $A_{1,r}$, a straightforward computation shows that 
$$i A_{1,r} = \widetilde{B}_r^* \widetilde{B}_r + F_r + T_r,\Longrightarrow a_{1,r} = \widetilde{b}_r^2 + f_r + t_r,$$
where $\widetilde{b}_r = \rho^{-1} b_r = j_r b$ is a symbol of order $m-1/2$ which arises when we differentiate $\chi_0$, with $\{B_r\}$ being a bounded family in $\Psi^s_b(M)$. The principal symbols $\{f_r\}$ are associated instead to the bounded family $\{F_r\}$ in $\Psi^s_b(M)$ which originates from the derivatives of $\chi_1$ while $t_r$ are principal symbols associated to the bounded family $\{T_r\}$ in $\Psi^{2s-1}_b(M)$, that includes the contribution by lower order terms. 

We choose the sought operator $Q$ as $\widetilde{B_r}$. In order to prove that $QA_ru$ is bounded in $\mathcal{L}^2(M)$ we analyze separately the usual two cases. We start from a boundary condition implemented by $\Theta\in\Psi^k_b(\partial M)$ with $k\leq 0$. In this case we can use \cite[Lemma 6.8]{GaWr18} with the due exception that one needs to replace in the proof Lemma 5.3 from \cite{GaWr18} with Lemma \ref{Lem: lemma-bound-quad-form-v1}.

\begin{lemma}\label{lemma:hyper-1}
There exist $C_1, c, \lambda, \delta_0 > 0$, a cutoff $\chi \in \mathcal{C}_0^\infty(M)$ and a compactly supported operator $G_2 \in \Psi^s_b(M)$ with 
$$ WF^\prime_b(G_2) \subset W \cap \{\ \zeta < 0 \} = \emptyset, $$
such that
\begin{equation}
\begin{split}
c \| \widetilde{B}_r u \|^2 \leq -2 Im \mathcal{E}_0(u, A^*_r A_r u) + C \Big( \| G_0 u \|_{\mathcal{H}^1(M)}^2 + \| G_1 u \|_{\mathcal{H}^1(M)}^2 + \\
+ \| G_2 u \|_{\mathcal{H}^1(M)}^2 + \| G_0 P_\Theta u \|^2_{\dot{\mathcal{H}}^{-1}(M)} + \|\chi u \|^2_{\mathcal{H}^{1,m}(M)} + \| \chi P_\Theta u \|^2_{\dot{\mathcal{H}}^{-1,m}(M)}  \Big).
\end{split}
\end{equation}
\end{lemma}

In the case where $\Theta\in\Psi^k_b(\partial M)$, with $0<k\leq 2$, we can exploit Lemma \ref{lemma-bound-quad-form-v2} in place of Lemma \ref{Lem: lemma-bound-quad-form-v1} to obtain the estimate

\begin{equation}
\begin{split}
c \| \widetilde{B}_r u \|^2 \leq -2 Im \mathcal{E}_0(u, A^*_r A_r u) + C \Big( \| G_0 u \|_{\mathcal{H}^1(M)}^2 + \| G_1 u \|_{\mathcal{H}^1(M)}^2 + \\
+ \| G_2 u \|_{\mathcal{H}^1(M)}^2 + \| G_0 P_\Theta u \|^2_{\mathcal{H}^{-1}(M)} + \|\chi u \|^2_{\mathcal{H}^{1,m}(M)} + \| \chi P_\Theta u \|^2_{\mathcal{H}^{-1,m}(M)} +\\
+  \| \chi \Theta u \|^2_{\mathcal{H}^{-1,m}(M)} + \| G_0 \Theta u \|^2_{\mathcal{H}^{-1}(M)}   \Big).
\end{split}
\end{equation}

At last we give a bound for $Im \mathcal{E}_0(u, A^*_r A_r u)$. As above we divide the analysis in two cases, starting from a boundary condition implemented by $\Theta\in\Psi^k_b(\partial M)$, with $0<k\leq 2$. 

\begin{lemma}\label{lemma:hyper-2}
Given $\varepsilon > 0$, there exists $\lambda > 0$ and $\delta_0 > 0$ such that
\begin{gather*}
Im \mathcal{E}_0(u,A^*_r A_r u) \leq\\
\varepsilon \| \widetilde{B}_r u \|^2_{\mathcal{H}^1(M)} + C \Big( \|G_2 u\|^2_{\mathcal{H}^1(M)} + \|G_0  P_\Theta u\|^2_{\dot{\mathcal{H}}^{-1}(M)} + \|G_0 \Theta u\|^2_{\dot{\mathcal{H}}^{-1}(M)} + \\
+  \| G_1 u \|^2_{\mathcal{H}^1(M)} + \|\chi u \|^2_{\mathcal{H}^{1,m}(M)} + \|\chi P_\Theta u \|^2_{\dot{\mathcal{H}}^{-1,m}(M)} + \|\chi \Theta u \|^2_{\dot{\mathcal{H}}^{-1,m}(M)} \Big),
\end{gather*}
for every $\delta \in (0,\delta_0)$.
\end{lemma}
\begin{proof}
Let $\Lambda_{-1/2} \in \Psi^{-1/2}_b(M)$ be an elliptic pseudodifferential operator. Then, there exists $\Lambda_{1/2} \in \Psi^{1/2}_b(M)$ such that $\Lambda_{1/2}\Lambda_{-1/2} = \mathbb{I} + R$ with $R \in \Psi^{-1}_b(M)$. In order to account for the boundary conditions, we bound $\mathcal{E}_\Theta(u,A^*_r A_r u)$. \begin{equation}\label{eq:abs-E_T}
\begin{split}
\big| \mathcal{E}_\Theta(u,A^*_r A_r u) \big| = & \big| \langle A_r P_\Theta u, A_r u \rangle \big|  =  \big| \langle A_r P_\Theta u, \Big( \Lambda_{1/2}\Lambda_{-1/2} + R \Big) A_r u \rangle \big| \\
\leq & \big| \langle A_r P_\Theta u, \Lambda_{1/2}\Lambda_{-1/2} A_r u \rangle \big| + \big| \langle A_r P_\Theta u, R A_r u \rangle \big|.
\end{split}
\end{equation}
We can control the first term similarly to the proof of Lemma \ref{Lem: lemma-bound-quad-form-v1}:
\begin{gather}
\big| \langle A_r P_\Theta u, \Lambda_{1/2}\Lambda_{-1/2} A_r u \rangle \big| \leq  C \Big( \| \Lambda_{1/2} A_r^* P\Theta u \|^2_{\dot{\mathcal{H}}^{-1}(M)} + \| \Lambda_{-1/2} A_r u\|^2_{\mathcal{H}^1(M)} \Big) \notag \\
\leq  C \Big( \| G_0 f \|^2_{\dot{\mathcal{H}}^{-1}(M)} + \| \chi f \|^2_{\dot{\mathcal{H}}^{-1,m}(M)} + \| G_1 u \|^2_{\mathcal{H}^1(M)} + \| \chi u \|^2_{\mathcal{H}^{1,m}(M)} \Big),
\end{gather}
where $G_0 \in \Psi_b^{s+1}(M)$ while $G_1 \in \Psi^s_b(M)$. Focusing on the second term of Equation \eqref{eq:abs-E_T}, we get
\begin{equation}
\begin{split}
\big| \langle A_r P_\Theta u, R A_r u \rangle \big| \leq \Big( \| G_0 f \|^2_{\dot{\mathcal{H}}^{-1}(M)} + \| \chi f \|^2_{\dot{\mathcal{H}}^{-1,m}(M)} +\\
\| G_1 u \|^2_{\mathcal{H}^1(M)} + \| \chi u \|^2_{\mathcal{H}^{1,m}(M)} \Big).
\end{split}
\end{equation}
The next step consists of finding a bound for 
$$|\mathrm{Im}\mathcal{E}_\Theta(u,A^*_r A_ru)-\textrm{Im}\mathcal{E}_0(u,A^*_r A_r u)|.$$
A direct inspection of Equation \eqref{Eq: Energy form with Theta} and of Equation \eqref{Eq: twisted Dirichlet form} unveils that this last difference consists of two terms. The first is
\begin{equation}
\langle A_r x^{-2} S_F u, A_r u \rangle - \langle A_r u, A_r x^{-2}S_F u \rangle,
\end{equation}
which can be rewritten as
\begin{equation}
\langle A^*_r [A_r,x^{-2} S_F] u, u \rangle - \langle u, A^*_r [A_r, x^{-2}S_F] u \rangle.
\end{equation}
Observing that $A^*_r [A_r, x^{-2}S_F]$ is uniformly bounded in $\Psi^{2s}_b(M)$, we find that
\begin{equation}
|\mathrm{Im} \langle x^{-2}S_F u, A^*_r A_r u \rangle| \leq C \Big( \| G_1 u \|^2_{\mathcal{H}^1(M)} + \| \chi u \|^2_{\mathcal{H}^{1,m}(M)}\Big).
\end{equation}
The second term is instead 
$$ 2 \mathrm{Im} \langle \Theta \gamma_- u, \gamma_- (A_r^* A_r u) \rangle_{\partial M},$$ 
which can be rewritten in the form

\begin{equation}\label{eq:diff_boundary}
\langle \Theta \gamma_- u, \gamma_- (A_r^* A_r u) \rangle_{\partial M} - \langle \gamma_- (A_r^* A_r u), \Theta \gamma_- u \rangle_{\partial M}.
\end{equation}

Proceeding as in the first bound of the proof and using the properties of the indicial operator as in Lemma \ref{Lem: lemma-bound-quad-form-v1}, we can write, modulo lower order terms
\begin{gather}
|\langle \Theta \gamma_- u, \gamma_- (A_r^* A_r u) \rangle_{\partial X} - \langle \gamma_- (A_r^* A_r u), \Theta \gamma_- u \rangle_{\partial M}| \leq \notag \\
\leq 2 |\langle \Theta \gamma_- u, \gamma_- (A_r^* \Lambda_{1/2} \Lambda_{-1/2} A_r u) \rangle_{\partial M} | \leq \notag \\
\leq \varepsilon \| \widetilde{B}_r u\|^2_{\mathcal{H}^1(M)} + C \Big( \|G_0 \Theta u \|_{\dot{\mathcal{H}}^{-1}(M)}^2+ \| \chi \Theta u \|_{\dot{\mathcal{H}}^{-1,m}(M)}^2 + \| G_1 u \|_{\mathcal{H}^1(M)}^2   \Big). 
\end{gather}
Collecting all estimates, we obtain the sought result.
\end{proof}

\noindent We focus on the case where $\Theta\in\Psi^k_b(\partial M)$ with $0<k\leq 2$.

\begin{lemma}\label{lemma:hyper-1b}
Given $\varepsilon > 0$, there exists $\lambda > 0$ and $\delta_0 > 0$ such that 
\begin{equation*}
\begin{split}
\mathrm{Im} \mathcal{E}_\Theta(u,A^*_r A_r u) \leq \varepsilon \| \widetilde{B}_r u \|^2_{\mathcal{H}^1(M)} + C \Big( \|G_2 u\|^2_{\mathcal{H}^1(M)} + \|G_0  P_\Theta u\|^2_{\dot{\mathcal{H}}^{-1}(M)} + \\
+ \| G_1 u \|^2_{\mathcal{H}^1(M)} + \|\chi u \|^2_{\mathcal{H}^{1,m}(M)} + \|\chi P_\Theta u \|^2_{\dot{\mathcal{H}}^{-1,m}(M)} \Big),
\end{split}
\end{equation*}
for every $\delta \in (0,\delta_0)$.
\end{lemma}

\begin{proof}
The first part of the proof is identical to that of Lemma \ref{lemma:hyper-2}. The difference lies in the estimates for the boundary terms, {\it cf. } Equation \eqref{eq:diff_boundary}.
This time, using the properties of the indicial family, {\it cf.} Equation \eqref{Eq: indicial family}, we can rewrite the relevant terms as \begin{equation*}
\begin{split}
\langle \widehat{N} \big( \widetilde{A}_r^* \widetilde{A}_r \Theta \big) (-i \nu_-) \gamma_- u, \gamma_- u \rangle_{\partial M} - \langle \widehat{N} \big( \Theta \widetilde{A}_r^* A_r \big) (-i \nu_-)   \gamma_- u, \gamma_- u \rangle_{\partial M} = \\
= \langle \widehat{N} [\widetilde{A}_r^* \widetilde{A}_r, \Theta](-i \nu_-) \gamma_- u, \gamma_- u \rangle_{\partial M}
\end{split}
\end{equation*}
where $\widetilde{A}_r = x^{2\nu_-} A_r x^{-2 \nu_-}$. Note that $A_r$ and $\widetilde{A}_r$ have the same principal symbol, hence we can write $\widetilde{A}_r^* = A_r + N_r$, with $N_r$ containing lower order terms. Hence: 
\begin{gather*}
[\widetilde{A}_r^* \widetilde{A}_r, \Theta] = [A^*_r A_r, \Theta] + [N^*_r A_r + A^*_r N_r + N_r^* N_r, \Theta] =\\
= [A^*_r A_r, \Theta] + \widetilde{A}^*_r \widetilde{A}_r - A^*_r A_r,
\end{gather*}
which yields
\begin{gather}
\langle [\widetilde{A}_r^* \widetilde{A}_r, \Theta] \gamma_- u, \gamma_- u \rangle_{\partial M} = \notag \\
= \langle \gamma_-([A^*_r A_r, \Theta]u),\gamma_- u \rangle_{\partial M} +
  \langle \Theta \gamma_- (\widetilde{A}^*_r\widetilde{A}_r - A^*_r A_r) u ], \gamma_- u \rangle_{\partial M}.\label{eq:hyp_inner_to_bound}
\end{gather}
We use Lemma \ref{lemma:commutator-positive} to control the first term writing $[A^*_r A_r, \Theta] = \widetilde{B}_r^* \widetilde{D}_r \widetilde{B}_r + T_r$ with $\widetilde{D}_r \in \Psi^k_b(M)$ and $T_r \in \Psi^{2s+k-1}_b(M)$. Observe that $ \widetilde{D}_r$ is related to $D_r\in\Psi^{k-2}_b(M)$ as in Lemma \ref{lemma:commutator-positive} since their respective principal symbols $d_r$ and $\widetilde{d}_r$ are connected via the identity $\widetilde{d}_r = \rho^2 d_r$ where $\rho=|\eta_{n-1}|$ as in Section \ref{Sec: Geometric preliminaries}. 
Hence
\begin{equation}\label{eq:hyp-boundary}
\begin{split}
\langle \gamma_-([A^*_r A_r, \Theta]u),\gamma_- u \rangle_{\partial M} =\notag \\
= \langle \gamma_-(\widetilde{D}_r B_r u), \gamma_-(\widetilde{B}_r u) \rangle_{\partial M} + \langle \gamma_- (T_r u), \gamma_- u \rangle_{\partial M},
\end{split}
\end{equation}
where, in the second equality, we used the properties of the indicial family to bring $\widetilde{B}^*_r$ to the right hand side. Thus it descends that
$$ \langle \gamma_-([A^*_r A_r, \Theta]u),\gamma_- u \rangle_{\partial M} = \langle \gamma_-(\widetilde{D}_r \widetilde{B_r} u), \gamma_-(\widetilde{B}_r u) \rangle_{\partial M},$$ 
modulo lower order terms bounded by $ s-1/2 $. Using the indicial family we obtain
\begin{gather*}
\langle \gamma_- (\widetilde{D}_r \widetilde{B}_r u), \gamma_- (\widetilde{B}_r u) \rangle_{\partial M} =\\
= \langle \widehat{N}(-i \nu_-)(\widetilde{D}_r) \widehat{N}(-i \nu_-)(\widetilde{B}_r)\gamma_- u, \widehat{N}(-i \nu_-)(\widetilde{B}_r) \gamma_- u\rangle_{\partial M}.
\end{gather*}
Using that $\widehat{N}(-i \nu_-)(\widetilde{D}_r) \in \Psi^0_b(M)$, $\widehat{N}(-i \nu_-)(\widetilde{B}_r) \in \Psi^{m-1/2}_b(M)$ together with Equation \eqref{prop-Psi0_Bound}, we obtain

\begin{equation}\label{eq:hyp-BDB-inner}
\begin{split}
|\langle \widehat{N}(-i \nu_-)(\widetilde{D}_r) \widehat{N}(-i \nu_-)(\widetilde{B}_r)\gamma_- u, \widehat{N}(-i \nu_-)(\widetilde{B}_r) \gamma_- u\rangle_{\partial M}| \leq \\
\leq \| \widehat{N}(-i \nu_-)(\widetilde{D}_r) \widehat{N}(-i \nu_-)(\widetilde{B}_r)\gamma_- u \|_{\mathcal{L}^2(M)}^2 + \| (\widetilde{B}_r)\gamma_- u \|_{\mathcal{L}^2(M)}^2 \leq \\
\| \chi \widehat{N}(-i \nu_-)(\widetilde{B}_r)\gamma_- u \|_{\mathcal{L}^2(M)}^2 + \| \widehat{N}(-i \nu_-)(\widetilde{B}_r)\gamma_- u \|_{\mathcal{L}^2(M)}^2 \leq 
\\
\leq \varepsilon \|\widetilde{B}_r u \|^2_{\mathcal{H}^1(M)}+ C (\|G_1 u\|_{\mathcal{H}^1(M)}+\|\chi u \|^2_{\mathcal{H}^{1,m}(M)}).
\end{split} 
\end{equation}

\noindent Focusing on the second term in Equation \eqref{eq:hyp_inner_to_bound}, we write 
\begin{eqnarray*}
\langle \Theta \gamma_- \big( (\widetilde{A}^*_r \widetilde{A}_r - A^*_r A_r) u \big), \gamma_- u \rangle_{\partial M} =\\
= \langle \Theta \gamma_- ( x^{2 \nu_-}[A^*_r A_r, x^{-2\nu_-}]u ), \gamma_- u \rangle_{\partial M}. 
\end{eqnarray*}
We can compute $[A^*_r A_r, x^{-2\nu_-}]$ thanks to Lemma \eqref{lemma:commutator-positive} obtaining
\begin{equation}
[A^*_r A_r, x^{-2\nu_-}] = (2 i \nu_-) \widetilde{B}^*_r \widetilde{B}_r + E_r + T_r.
\end{equation} 
Each term can be controlled as above, obtaining ultimately
\begin{equation}
\begin{split}
|\langle \Theta \gamma_- \big( (\widetilde{A}^*_r \widetilde{A}_r - A^*_r A_r) u \big), \gamma_- u \rangle_{\partial X}| \leq \varepsilon \| \widetilde{B}_r u \|^2_{\mathcal
{H}^1} +\\
C \Big( \| G_1 u \|^2_{\mathcal{H}^{1}(M)} + \| \chi u \|^2_{\mathcal{H}^{1,m}(M)} + \| C_2 u \|^2_{\mathcal{H}^1(M)} \Big),
\end{split}
\end{equation} 
which entails the sought conclusion.
\end{proof}

Finally we can complete the proofs of Propositions \ref{prop:hyp_reg_pos} and \ref{prop:hyp_reg_negnull}. Here we focus only on the first case since the second one follows suit.

\vskip .3cm

\noindent{\em Proof of Proposition \ref{prop:hyp_reg_pos}:} We sketch the main steps since we can proceed exactly as in the elliptic case, {\it cf.} Proposition \ref{Prop: elliptic regularity}. Most notably we follow an induction procedure with respect to $s$. Notice in particular that the statement holds true for $s<m+\frac{1}{2}$ since $u\in\mathcal{H}^{1,m}_{loc}(M)$. To continue in the inductive procedure we consider once more a family $J_r\in\Psi_b^{m-s-1}(M)$, $r\in(0,1)$, such that $J_r \rightarrow \mathbb{I} \in \Psi^0_b(M)$.  Then $\widetilde{B}_r \rightarrow  \widetilde{B} \in \Psi^s_b(M)$ as $r \rightarrow 0$. Using Lemma \ref{lemma:hyper-1} and Lemma \ref{lemma:hyper-2} or \eqref{lemma:hyper-1b} depending on the order of $\Theta$, one obtains that $\| \widetilde{B}_r u \|_{\mathcal{H}^1(M)}$ is uniformly bounded. Therefore, we can find a subsequence $\widetilde{B}_{r_k}u$, with $r_k \rightarrow 0$ for $k \rightarrow +\infty$, that is weakly convergent in $\mathcal{H}^1(M)$. Since $\widetilde{B}_r  \rightarrow \widetilde{B} u$ in $\mathcal{D}^\prime(M)$, the weak limit lies in $\mathcal{H}^1(K)$ for a suitable compact subset $K \subset M$. By uniqueness of the limit and considering that $\widetilde{B}$ is elliptic at $q_0$, we obtain the thesis. 

\subsection{Estimates in the glancing region}\label{Sec: glancing region}

At last we focus on the glancing region $\mathcal{G}(M)$ as in Equation \eqref{Eq: glancing region}. As in the previous subsection, we use a positive commutator argument to obtain the sought microlocal estimates. Barring some geometrical aspects we proceed similarly to Propositions \ref{prop:hyp_reg_pos} and \ref{prop:hyp_reg_negnull}. For this reason, we introduce in some details mainly the geometric framework. In the following $U$ will denote an open coordinate neighbourhood, while $V = U \cap \partial M\neq\emptyset$. As in the previous section we need to consider two scenarios depending on the class of boundary conditions, namely $\Theta\in\Psi^k_b(\partial M)$ with either $k\leq 0$ or $0<k\leq 2$. Similarly to the preceding cases, we shall pick the trivial extension of $\Theta$ to $M$, indicating it with the same symbol. In the following $y_{n-1}$ still refers to the time coordinate corresponding to $\tau$ in Theorem \ref{Th: globally hyperbolic}, while $\eta_{n-1}$ is the corresponding momentum on the $b$-cotangent bundle. In addition $q_0$ refers to a point lying in a compact region $K$ where 

$$K \subset (\mathcal{G} \cap T^*_U M )\setminus WF_b^{-1,s+1}(P_\Theta u)\quad\textrm{if}\quad k\leq 0,$$ 
or 
$$K \subset (\mathcal{G} \cap T^*_U M )\setminus \left(  WF_b^{-1,s+1}(P_\Theta u) \cup  WF_b^{-1,s+1}(\Theta u) \right) \quad\textrm{if}\quad 0<k\leq 2.$$ 
In local coordinates $q_0$ reads $(0, (y_0)_i, 0, (\eta_0)_i)$, $i=1,\dots, n-1$, while it holds $\widehat{g}^{ij}(0,y_0) (\eta_0)_i (\eta_0)_j = 0$.  Since $\eta_{n-1} \neq 0$, we can use the projective coordinates on ${}^b S^* M$ near $\pi_S(q_0)$, where $\pi_S : {}^b T^*M\setminus\{0\} \rightarrow {}^b S^*M$ is the quotient map.
We denote the projection to the boundary with 
$$\widetilde{\pi}: T^*_U M \rightarrow T^*V,$$ 
$$(x,y_i,\xi,\eta_i) \mapsto (y_i,\eta_i),$$
where $i=1,\dots,n-1$. As last ingredient we introduce the gliding vector field $W$, describing the evolution of a point in the directions tangent to the boundary. Consider thus a point on $\widetilde{\pi}(T^*_UM)$ of coordinates $(\eta_0,(y_0)_i)$, $i=1,\dots,n-1$ and define
\begin{equation}
W(\eta_0,(y_0)_i) = \sum_{i = 1}^{n-1}(\partial_{\eta_i} \widehat{p}) (0,\eta_0,0,(y_0)_i) \partial_{y_i} - (\partial_{y_i} \widehat{p})(0,\eta_0,0,(y_0)_i) \partial_{\eta_i},
\end{equation}
where $\widehat{p}$ is the principal symbol of $x^{-2}P$, $P$ being the Klein-Gordon operator as in Equation \eqref{Eq: KG on M}. Letting $\rho = |\eta_{n-1}|$, we observe that, in a neighbourhood of $(\eta_0,(y_0)_i)$, $\rho^{-1}W$ is a non degenerate vector field, since $\rho^{-1} W y^{n-1} = 2 sgn(\eta_{n-1})$. Thus we can use the straightening theorem \cite{LPV13} to find $2n-2$ homogeneous degree zero functions $\rho_1, \cdots, \rho_{2n-2} \in T^*_U M$ with linearly independent differentials such that $\rho^{-1} W \rho_1 = 1$ and $\rho^{-1} W \rho_i = 0$ for $i=2,\cdots,2n-2$.  We also note that $\widehat{p}(0,\eta,0,y_i)$ is annihilated by $W$. Since $d\widehat{p} \neq 0$, we can set $\rho_2(\eta_0,(y_0)_i) = \widehat{p}(0,\eta_0,0,(y_0)i)$. Then we extend $\rho_1,\cdots,\rho_{2n-2}$ in such a way to be independent from $(x,\xi)$, in order to obtain a local chart whose coordinate functions are $x,\widehat{\zeta},\rho_1,\cdots,\rho_{2n-2}$. 

With these data we can introduce two homogeneous functions $\omega_0$ and $\omega$ over $K \cap V$, playing the same role as $\eta$ and $\omega$ in the hyperbolic region:
$$ \Omega_0 = \sum_{i=1}^{2n-2} (\rho_i - \rho_i(q_0))^2, \ \ \ \ \Omega = x^2 + \Omega_0,$$
where we omit to indicate the explicit dependence on $q_0$ for the sake of simplicity of the notation. In connection to these functions we introduce  
$$ \phi_0 = \rho_1 + \frac{\Omega_0}{\lambda^2 \delta}, \quad \phi = \rho_1 + \frac{\Omega}{\lambda^2 \delta}. $$

Using the same cutoff functions $\chi_0$ and $\chi_1$ introduced in Section \ref{Sec: hyperbolic region}, we localize near $q_0$ using a b-pseudodifferential operator $A$ of order zero whose total symbol is given by 
$$ a = \chi_0(2-\phi/\delta)\chi_1(1+(\rho+\delta)/(\lambda\delta)). $$
The ensuing families $\mathcal{A}$, $\mathcal{B}$ and $\widetilde{\mathcal{B}}$ are defined as in Section \ref{Sec: hyperbolic region}. With these data, the following generalizations of \cite[Prop. 6.11]{GaWr18} hold true. Observe that, with a slight abuse of notation, we identify subsets of the b-cosphere bundle with their pre-image on the b-cotangent bundle.

\begin{prop}\label{prop:glan_reg_negnulla}
Let $\Theta \in \Psi_b^k(M)$ with $k \leq 0$ and let $u \in \mathcal{H}^{1,m}_{loc}(M)$ with $m \leq 0$. If $K \subset {}^b S^*_U M$ is compact and $K \subset (\mathcal{G} \cap T^*\partial M) \setminus WF_b^{-1,s+s}(P_\Theta u)$, then there exist $C_0, \delta_0 > 0$ such that for each $q_0 \in K$ and $\delta \in (0,\delta_0)$ the following holds. Let $\alpha_0 \in \mathcal{N}$ be such that $\pi(\alpha_0) = q_0$. If 
$$ \alpha \in \mathcal{N}, \ \ |\widetilde{\pi}(\alpha) - exp(-\delta W)(\widetilde{\alpha}_0)| \leq C_0 \delta^2, \ \ \ |x(\alpha)| \leq C_0 \delta^2, $$
imply $\pi(\alpha) \not \in WF_b^{1,s}(u)$, then $q_0 \not \in WF_b^{1,s}(u)$.
\end{prop}

In the case where $\Theta\in\Psi^k_b(M)$, $0<k\leq 2$, the generalization of \cite[Prop. 6.11]{GaWr18} is the following:

\begin{prop}\label{prop:glan_reg_pos}
Let $\Theta \in \Psi_b^k(M)$ with $0 < k \leq 2$ and let $u \in \mathcal{H}^{1,m}_{loc}(M)$ with $m \leq 0$. If $K \subset {}^b S^*_U X$ is compact and 
$$K \subset (\mathcal{G} \cap T^*\partial M) \setminus \left(  WF_b^{-1,s+s}(P_\Theta u)  \cup WF_b^{-1,s+s}(\Theta u)   \right), $$ 
then there exist $C_0, \delta_0 > 0$ such that for each $q_0 \in K$ and $\delta \in (0,\delta_0)$ the following holds. Let $\alpha_0 \in \mathcal{N}$ be such that $\pi(\alpha_0) = q_0$. If
$$ \alpha \in \mathcal{N}, \ \ |\widetilde{\pi}(\alpha) - exp(-\delta W)(\widetilde{\alpha}_0)| \leq C_0 \delta^2, \ \ \ |x(\alpha)| \leq C_0 \delta^2, $$
imply $\pi(\alpha) \not \in WF_b^{1,s}(u)$, then $q_0 \not \in WF_b^{1,s}(u)$.
\end{prop}

We focus on the case $0<k\leq 2$, the other following suit. The proof is based on two lemmas along with the counterpart of Lemma \ref{lemma:hyper-1} for the glancing region. The proofs are similar to those of the hyperbolic case and are adaptation of those in \cite{GaWr18}, hence we will omit them. 

The first lemma we need gives a bound of the difference between the $\mathcal{L}^2$-norm of $Q_0 A_r u$ and of a generic positive sesquilinear $\mathcal{Q}$ applied $A_r u$.
\begin{lemma}
Let $U \subset M$ be a boundary coordinate patch and $m \leq 0$. Let $\mathcal{A}=\{A_r : r \in (0,1) \}$ be a bounded subset of $\Psi^s_b(M)$ with compact support in $U$ such that $A_r \in \Psi^m_b(M)$ for each $r \in (0,1)$. Let $\delta>0$ and let $V_\delta=\{ q \in {}^bT^*_UM \setminus\{0\}: \widehat{g}^{ij} \eta_i \eta_j \leq \delta \beta^{-1}|\eta_{n-1}|^2 \}$ and assume that $WF_b^\prime(\mathcal{A}) \subset V_\delta$. Let $G_0 \in \Psi^s_b(M)$ and $G_1 \in \Psi^{s-1/2}_b(M)$ be elliptic on $WF^\prime_b(\mathcal{A})$ and on $WF^\prime_b(\Theta \mathcal{A})$ respectively, both with compact support in $U$. Then there exist $C_\varepsilon$ and $\chi \in \mathcal{C}_0^\infty(U)$ such that
\begin{gather*}
\|Q_0 A_r u\|^2_{\mathcal{L}^2(M)} - \varepsilon \mathcal{Q}(A_r u, A_r u) \leq \\
\leq 2 \delta \| Q_{n-1} A_r u \|^2_{\mathcal{L}^2(M)} + C_\varepsilon \Big( \| \chi u \|^2_{\mathcal{H}^{1,m}(M)} + \| \chi P_\Theta u \|^2_{\dot{\mathcal{H}}^{-1,m}(M)} +\\
+ \| G_0 P_\Theta u \|^2_{\dot{\mathcal{H}}^{-1}(M)} +
 \| G_1 u \|^2_{\mathcal{H}^1(M)} +  \| \chi \Theta u \|^2_{\dot{\mathcal{H}}^{-1,m}(M)} + \| G_0 \Theta u \|^2_{\dot{\mathcal{H}}^{-1}(M)}  \Big).
\end{gather*}
\end{lemma}

The proof follows that of Lemma 6.10 in \cite{GaWr18} up to the fact that we use Lemma \eqref{lemma-bound-quad-form-v2} to control the boundary terms. This result can be used to generalize straightforwardly the proof of Lemma 6.12 in \cite{GaWr18} to the case in hand:

\begin{lemma}
There exist $C, c, \lambda, \delta_0 > 0$, a cutoff $\chi \in \mathcal{C}^\infty_0(M)$, $G_0,G_1$ as above and an operator $G_2 \in \Psi^s_b(M)$ with
$$WF_b^\prime(G_2)\subset W \cap \{ -2\delta \lambda < \rho_1 < -\delta/2, \omega^{1/2} < 3 \lambda \delta \},$$
such that 
\begin{gather*}
c\|\widetilde{B}_r u \|^2_{\mathcal{H}^1(M)} \leq\\
\leq - 2 Im \mathcal{E}_0(u, A^*_r A_r u) + C \Big(  \|G_1 u \|^2_{\mathcal{H}^1(M)} + \|G_2 u \|^2_{\mathcal{H}^1(M)} + \|\chi u \|^2_{\mathcal{H}^{1,m}(M)} + \\
\|G_0 P_\Theta u \|^2_{\dot{\mathcal{H}}^{-1}(M)}  + \|\chi P_\Theta u \|^2_{\dot{\mathcal{H}}^{-1,m}(M)}  +  \| \chi \Theta u \|^2_{\dot{\mathcal{H}}^{-1,m}(M)} + \| G_0 \Theta u \|^2_{\dot{\mathcal{H}}^{-1}(M)}  \Big).
\end{gather*}
\end{lemma}

Proposition \ref{prop:glan_reg_pos} can now be proven along the lines of Proposition \ref{Prop: elliptic regularity} and of Proposition \ref{prop:hyp_reg_pos}, the differences arising because of the geometric nature of the glancing region. We refer to \cite[Sec. 6]{GaWr18} and to \cite[Sec. 7]{Vasy08} for further details.

\subsection{Propagation of singularities theorems}\label{Sec: propagation of singularities theorem}

Combining all microlocal estimates from the previous sections, we obtain the following propagation of singularities theorem which generalizes that of \cite{GaWr18}. We recall that, once more, with $\Theta$ we denote also the trivial extension to $M$ of the pseudodifferential operator on $\partial M$ which implements the boundary condition. 

\begin{thm}\label{Thm: main theorem k positivo}
Let $\Theta \in \Psi_b^k(\partial M)$ with $0<k\leq 2$. If $u \in \mathcal{H}_{loc}^{1,m}(M)$ for  $m \leq 0$ and $s \in \mathbb{R} \cup \{ + \infty \}$, then $WF_b^{1,s}(u) \setminus \left( WF_b^{-1,s+1}(P_\Theta u) \cup WF_b^{-1,s+1}(\Theta u) \right) $ is the union of maximally extended generalized broken bicharacteristics within the compressed characteristic set $\dot{N}$.
\end{thm}

\noindent In full analogy it holds also

\begin{thm}\label{Thm: main theorem k negativo}
Let $\Theta \in \Psi_b^k(M)$ with $k \leq 0$. If $u \in \mathcal{H}_{loc}^{1,m}(M)$ for $m \leq 0$ and $s \in \mathbb{R} \cup \{ + \infty \}$, then it holds that $WF_b^{1,s}(u) \setminus WF_b^{-1,s+1}(P_\Theta u)$ is the union of maximally extended GBBs within the compressed characteristic set $\dot{N}$.
\end{thm}

The proof of both theorems is similar to that given in \cite{Vasy08}, employing the estimates derived in the previous sections. Hence we do not give all details here, rather we feel worth outlining only the analysis of forward propagation in the hyperbolic region for the reader's convenience. We focus on the case of Theorem \ref{Thm: main theorem k positivo}.

In $\mathring{M}$, the statement can be reduced to Duistermaat and H\"{o}rmander's theorem of propagation of singularities \cite{HoDu72, Hor71}. We can focus on the boundary. First, we prove a local version of the theorem, extending a non maximal GBB to $\partial M$, {\it cf.} Definition \ref{Def: generalized broken bicharacteristics}. In other words, calling $U$ an open chart of $M$ such that $U\cap\partial M\neq\emptyset$, we show that if $q_0 \in WF_b^{1,s}(u) \setminus WF_b^{-1,s+1}(P_\Theta u) $ with $q_0 \in {}^bT^*_U M$, then there exists a GBB $\gamma:[-\varepsilon_0 , 0] \rightarrow \dot{N}$, with $\varepsilon_0>0$, such that $\gamma(0) = q_0$ and $\gamma(s) \in WF_b^{1,s}(u) \setminus\left(WF_b^{-1,s+1}(P_\Theta u)\cup WF_b^{-1,s+1}(\Theta u)\right) $ for $s \in [-\varepsilon_0,0]$.  We consider the case in which $q_0 \in \mathcal{H}(M)$ and that in which $q_0 \in \mathcal{G}(M)$ separately.

We focus on the former proceeding iteratively. Given $q_0 \in {}^bT^*_U M$, we build a sequence of generalized broken bicharacteristics $\gamma_j:[-\varepsilon_0,0] \rightarrow \dot{N}$ such that $\gamma_j(s) \in WF_b^{1,s}(u) \setminus\left(WF_b^{-1,s+1}(P_\Theta u)\cup WF_b^{-1,s+1}(\Theta u)\right)$ and with the endpoint $\gamma_j(0) \doteq q_j \in {}^bT^* \mathring{M} $ converging to $q_0$ on the boundary.
Thanks to Proposition \ref{prop:hyp_reg_pos}, choosing increasingly smaller sets $W\subset T^*M\setminus\{0\}$ we can found the sought sequence of points $\{q_j\}_{j\in\mathbb{N}}$. Since every $q_j\in\mathring{M}$, H\"{o}rmander's theorem on propagation of singularities \cite{Hor3} \cite{HoDu72} guarantees existence of the sought sequence of GBBs. The assumption of forward propagation, that is $\xi(q_j) <0 $, ensures that there exists $\varepsilon_0>0$ such that for $s \in [-\varepsilon_0,0]$, $\gamma_j(s) \not \in {}^b T^*_Y M$, where $\xi = 0$. 

Since generalized broken bicharacteristics $\mathcal{R}_K[-\varepsilon_0,0]$ with $K$ compact are themselves compact in the topology of uniform convergence, Lemma \ref{lemm:gbb-lebau} allows to conclude that there exists a subsequence $\{\gamma_{j_k}\}$ uniformly converging to $\gamma$.  

At last we extend the result to maximal GBBs. Given a subset $V \subset \dot{N}$ with $q \in V$ and $a,b \in \mathbb{R}$ containing $0$, there is a natural partial order on the set $\mathcal{GBB}_q$ of broken generalized bicharacteristics $\gamma:(a,b)\rightarrow V$ such that $\gamma(0)=q$. Let $\gamma_1:(a_1,b_1) \rightarrow V$ and $\gamma_2: (a_2,b_2) \rightarrow V$ be two elements of $\mathcal{GBB}_q$, we say that $\gamma_1 \leq \gamma_2$ if $(a_1,b_1) \subset (a_2,b_2)$ and if the two curves agree over the common domain $(a_1,b_1)$. 

Since a non-empty totally ordered subset has an upper bound, we can extend the GBBs joining the domains of those in the chain. At this point we apply Zorn's lemma, the maximal element of any totally ordered subsets being the maximal extension of a GBBs. In the glancing region the main idea is still to build a sequence of curves approximating a GBB, although the details are different due to some technical hurdles. The reader can find the argument in \cite{Leb97} and \cite{Vasy08}.

\section*{Acknowledgments}

We are grateful to Benito Juarez Aubry for the useful discussions which inspired the beginning of this project and to Nicolò Drago both for the useful discussions and for pointing out reference \cite{GMP}. The work of A. Marta is supported by a fellowship of the Università Statale di Milano, which is gratefully acknowledged. C. Dappiaggi is grateful to the Department of Mathematics of the Università Statale di Milano for the kind hospitality during the realization of part of this work.


\begin{thebibliography}{99} 

\bibitem[AD99]{Ashtekar:1999jx}
A.~Ashtekar and S.~Das,
\textit{``Asymptotically Anti-de Sitter space-times: Conserved quantities,''}
Class.\ Quant.\ Grav.\  \textbf{17} (2000), L17-L30
[arXiv:hep-th/9911230 [hep-th]].
	
\bibitem[AFS18]{AFS18}
L.~Ak\'{e} Hau, J.~L.~Flores, M.~S\'{a}nchez,
\textit{``Structure of globally hyperbolic spacetimes with timelike boundary''},
arXiv:1808.04412 [gr-qc], to appear in Rev. Mat. Iberoamericana (2020).

\bibitem[AGN16]{ANN16}
B.~Ammann, N.~Gro{\ss}e and V. Nistor,
\textit{``Well-posedness of the Laplacian on manifolds with boundary and bounded geometry''}, Math. Nachr. {\bf 292} (2019) 1213.
arXiv:1611.00281 [math-AP].

\bibitem[Bac11]{Bachelot:2010zw}
A.~Bachelot,
\textit{``The Klein-Gordon Equation in Anti-de Sitter Cosmology,''}
J. Math. Pure. Appl. \textbf{96} (2011), 527
[arXiv:1010.1925 [math-ph]].

\bibitem[BF82]{BF82}
P.~Breitenlohner, D.~Z.~Freedman,
\textit{``Stability in gauged extended supergravity''},
Annals Phys. {\bf 144}, (1982) 249.

\bibitem[Coc14]{Coclite:2014} 
G.~M. Coclite, et al,
\textit{``Continuous dependence in hyperbolic problems with Wentzell boundary conditions,''}
Commun. Pure Appl. Anal. {\bf 13} (2014), 419.

\bibitem[DDF18]{DDF18}
C.~Dappiaggi, N.~Drago, H.~Ferreira 
\textit{``Fundamental solutions for the wave operator on static Lorentzian manifolds with timelike boundary''}, Lett. Math. Phys. \textbf{109} (2019), 2157, [arXiv:1804.03434 [math-ph]].

\bibitem[DF17]{Dappiaggi:2017wvj}
C.~Dappiaggi and H.~R.~C.~Ferreira,
\textit{``On the algebraic quantization of a massive scalar field in anti-de-Sitter spacetime,''}
Rev. Math. Phys. \textbf{30} (2017) no.02, 1850004,
[arXiv:1701.07215 [math-ph]].

\bibitem[DFJ18]{Dappiaggi:2018pju}
C.~Dappiaggi, H.~R.~Ferreira and B.~A.~Juárez-Aubry,
\textit{``Mode solutions for a Klein-Gordon field in anti–de Sitter spacetime with dynamical boundary conditions of Wentzell type,''}
Phys. Rev. D \textbf{97} (2018) no.8, 085022
[arXiv:1802.00283 [hep-th]].

\bibitem[DW19]{Dybalski:2018egv}
W.~Dybalski and M.~Wrochna,
\textit{``A mechanism for holography for non-interacting fields on anti-de Sitter spacetimes,''}
Class. Quant. Grav. \textbf{36} (2019) no.8, 085006
[arXiv:1809.05123 [math-ph]].

\bibitem[DuH\"o72]{HoDu72}
J.~Duistermaat, L.~H\"{o}rmander,
\textit{Fourier integral operators II},
Acta Mathematica {\bf 128}, (1972) 183.

\bibitem[EnKa13]{Enciso:2013lza}
A.~Enciso and N.~Kamran,
\textit{``A singular initial-boundary value problem for nonlinear wave equations and holography in asymptotically anti-de Sitter spaces,''}
J. Math. Pure. Appl. \textbf{103} (2015), 1053
[arXiv:1310.0158 [math.AP]].

\bibitem[FGGR02]{Favini:2002} 
A.~Favini, G.R. Goldstein, J.A. Goldstein, S. Romanelli,
\textit{``The heat equation with generalized wentzell boundary condition,''}
J. Evol. Equ. {\bf2} (2002), 1.

\bibitem[Gan18]{Gan15}
O.~Gannot,
\textit{``Elliptic boundary value problems for Bessel operators, with applications to anti-de Sitter spacetimes''},
Comptes Rendus Mathematique {\bf 356} (2018) 988.

\bibitem[GOW17]{GOW17}
C.~G\'{e}rard, O.~Oulghazi, M.~Wrochna,
\textit{``Hadamard States for the Klein-Gordon equation on Lorentzian manifolds of bonded geometry''},
Comm. Math. Phys. \textbf{352} (2017) 519
[arXiv:1602.00930 [math-ph]].

\bibitem[GMP14]{GMP}
V. Guillemin, E. Miranda and A. R. Pires
\textit{``Symplectic and Poisson geometry on b-manifolds''}
Adv. in Math. {\bf 264} (2014) 864.	
arXiv:1206.2020 [math.SG].

\bibitem[Gri00]{Gri00}
D.~Grieser
\textit{``Basics of the b-Calculus''} in Approaches to Singular Analysis, Advances in Partial Differential Equations 125,  J. B. Gil, D. Grieser, M. Lesch (Eds.), Birkh\"auser, Basel, 2001, 
arXiv:math/0010314 [math.AP].

\bibitem[GS13]{GS13}
N.~Gro{\ss}e, C.~Schneider,
\textit{``Sobolev spaces on Riemannian manifolds with bounded geometry:General coordinates and traces''},
Math. Nachr. {\bf 286} (2013) 1586.

\bibitem[GW18]{GaWr18}
O.~Gannot, M.~Wrochna
\textit{``Propagation of Singularities on AdS Spacetimes for General Boundary Conditions ant the Holographic Hadamard Condition''}, to appear on J. Inst. Math. Juissieu (2020), arXiv:1812.06564 [math.AP].

\bibitem[Hol12]{Hol12}
G.~Holzegel, 
\textit{``Well-posedness for the massive wave equation on asymptotically anti-de
Sitter spacetimes''}, 
J. Hype. Diff. Eq. {\bf 9} (2012), 239.

\bibitem[H\"or00]{Hor3}
L.~H\"{o}rmander
\textit{``The analysis of partial differential operators III''}
Springer (2000), 525 pg. 

\bibitem[H\"or03]{Hor1} 
L.~H\"{o}rmander
\textit{``The Analysis of Linear Partial Differential Operators I''}
Springer-Verlag (2003), 440 pg.

\bibitem[H\"or71]{Hor71}
L.~H\"{o}rmander, 
\textit{On the existence and the regularity of solutions of linear pseudo-differential equations},
L'Enseignement Math\'ematique Vol. \textbf{18} (1971), 69 pg.

\bibitem[IW03]{WaIs03}
A.~Ishibashi, R.~M.~Wald,
\textit{``Dynamics in Non-Globally-Hyperbolic Static Spacetimes II: General Analysis of Prescriptions for Dynamics''},
Class. Quant. Grav. {\bf 20} (2003) 3815, [arXiv:gr-qc/0305012 [gr-qc]].

\bibitem[Jos99]{Jos99}
M.~S.~Joshi 
\textit{``Lectures on Pseudo-differential Operators''}
arXiv:math/9906155 [math.AP]

\bibitem[L\"am93]{Lam93}
C.~L\"{a}mmerzahl,
\textit{The pseudodifferential operator square root of the Klein–Gordon equation},
J. Math. Phys. {\bf 34}, (1993) 3918.

\bibitem[Leb97]{Leb97}
G.~Lebau,
\textit{Propagation des ondes dans les vari\`et\`es \'a coins}
Ann. Scient. \`Ecole Norm. Sup. \textbf{30} (1997), 429. 

\bibitem[LPV13]{LPV13}
C.~Laurent-Gengoux, A.~Pichereau, P.~Vanhaecke. 
\textit{``Poisson Structures''},
Springer (2013), Comprehensive Studies in Mathematics, 464 pg.

\bibitem[McSa11]{McSa11}
P.~Mckeag, Y.~Safarov,
\textit{Pseudodifferential operators on manifolds: a coordinate-free approach},
In `Partial Differential Equations and Spectral Theory (Operator Theory: Advances and Applications)', Vol. {\bf 211}, Birkh\"{a}user (2011), Eds.  M.~Demuth, B.~Schulze, I.~Witt. arXiv:1106.3637 [math.AP].

\bibitem[Mel92]{Mel92}
R.~B.~Melrose,
\textit{``The Atiyah-Patodi-Singer index theorem''},
Research Notes in Mathematics, (1993) CRC Press, 392pg.

\bibitem[Mel81]{Mel81}
R.~B.~Melrose,
\textit{``Transformation of Boundary Problems''}, 
Acta Math. {\bf 147} (1981), 149.

\bibitem[MP92]{MePi92}
R.~Melrose, P.~Piazza
\textit{`Analytic K-theory on manifolds with corners'} 
Elsevier, Advances in Mathematics {\bf 92}, (1992) 1.

\bibitem[Mel93]{Mel93}
R.~B.~Melrose,
\textit{``The Atiyah-Patodi-Singer Index Theorem''}, 
A K Peters/CRC Press (1981), 392pg.

\bibitem[ON83]{ONeill83}
B.~O’Neill, 
\textit{``Semi-Riemannian Geometry with Applications to Relativity''},
San Diego Academic Press (1983), 468pg.

\bibitem[Sch01]{Sch01}
T.~Schick,
\textit{``Manifolds with boundary and of bounded geometry''},
Math. Nachr. {\bf 223} (2001) 103, arXiv:math/0001108 [math.DG].

\bibitem[War13]{War1}
C.~M.~Warnick
\textit{``The massive wave equation in asymptotically AdS spacetimes''},
Comm. Math. Phys, {\bf 321} (2013) 85.

\bibitem[Ue73]{Ueno:1973} 
T.~Ueno,
\textit{``Wave equation with Wentzell's boundary condition and a related semigroup on the boundary, I,''}
Proc. Japan Acad. {\bf 49} (1973), 672.

\bibitem[Vas08]{Vasy08}
A.~Vasy
\textit{``Propagation of singularities for the wave equation on manifolds with corners''},
Annals of Mathematics, 168 (2008), 749,
arXiv:math/0405431 [math.AP].

\bibitem[Vas10]{Vasy10}
A.~Vasy
\textit{``Diffraction at corners for the wave equation on differential forms''},
Comm. Part. Diff. Eq. {\bf 35} (2010), 1236, arXiv:0906.0738 [math.AP] 

\bibitem[Vas12]{Vasy12}
A.~Vasy
\textit{``The wave equation on asymptotically Anti-de Sitter spaces''},
Analysis \& PDE {\bf 5} (2012), 81, arXiv:0911.5440 [math.AP].

\bibitem[Za15]{Zahn:2015due}
J.~Zahn,
\textit{``Generalized Wentzell boundary conditions and quantum field theory,''}
Annales Henri Poincare \textbf{19} (2018) no.1, 163-187
[arXiv:1512.05512 [math-ph]].

\end{thebibliography}
\end{document}